\documentclass[11pt]{article} 

\usepackage[utf8]{inputenc}
\usepackage[pdftex,left=1in,top=1in,bottom=1in,right=1in]{geometry}

\usepackage{amsmath, amssymb, dsfont, mathrsfs,amsthm}

\usepackage{graphicx}
\usepackage{caption}
\usepackage{bm}
\usepackage{tikz}
\usepackage{pgfplots}
\usepackage{xcolor}
\usetikzlibrary{positioning}
\usepackage[colorlinks=true,citecolor=blue,linkcolor=blue]{hyperref}
\hypersetup{
    colorlinks,
    linkcolor={red!50!black},
    citecolor={blue!50!black},
    urlcolor={blue!80!black}
}
\usepackage[nameinlink, noabbrev, capitalize]{cleveref}

\newcommand{\ap}{\Phi}

\newcommand{\cC}{\mathcal{C}}

\newcommand{\cB}{\mathcal{B}}

\newcommand{\cS}{\mathcal{S}}

\newtheorem{thm}{Theorem}

\newtheorem{lem}{Lemma}

\newtheorem{defn}{Definition}

\newcommand{\bits}{\{0,1\}}

\newcommand{\Enc}{\mathsf{Enc}}
\newcommand{\Dec}{\mathsf{Dec}}

\definecolor{mygreen}{HTML}{588D6A}
\definecolor{myred}{HTML}{C86733}
\definecolor{myblue}{HTML}{5B68FF}
\definecolor{myshadow}{HTML}{E6C5B4}

\let\originalleft\left
\let\originalright\right
\renewcommand{\left}{\mathopen{}\mathclose\bgroup\originalleft}
\renewcommand{\right}{\aftergroup\egroup\originalright}

\title{Beyond Single-Deletion Correcting Codes:\\ Substitutions and Transpositions
}

\author{Ryan Gabrys\thanks{Department of Electrical and Computer Engineering, University of California, San Diego, USA. Email: \texttt{ryan.gabrys@gmail.com}.}\and Venkatesan Guruswami\thanks{Computer Science Department, Carnegie Mellon University, Pittsburgh, USA. Email: \texttt{venkatg@cs.cmu.edu}. Research supported  in part by the NSF grants CCF-1814603 and CCF-2107347.}\and Jo\~ao Ribeiro\thanks{Computer Science Department, Carnegie Mellon University, Pittsburgh, USA. Email: \texttt{jlourenc@cs.cmu.edu}. Research supported in part by the NSF grants CCF-1814603 and CCF-2107347 and by the following grants of Vipul Goyal: the NSF award 1916939, DARPA SIEVE program, a gift from Ripple, a DoE NETL award, a JP Morgan Faculty Fellowship, a PNC center for financial services innovation award, and a Cylab seed funding award.}\and Ke Wu\thanks{Computer Science Department, Carnegie Mellon University, Pittsburgh, USA. Email: \texttt{kew2@cs.cmu.edu}. Research supported in part by a DARPA SIEVE award, SRI Subcontract Number 53978, and DARPA Prime Contract Number HR00110C0086.}
}

\date{}

\begin{document}
	
\maketitle
	
\begin{abstract}
We consider the problem of designing low-redundancy codes in settings where one must correct deletions in conjunction with substitutions or adjacent transpositions; a combination of errors that is usually observed in DNA-based data storage.
One of the most basic versions of this problem was settled more than 50 years ago by Levenshtein, who proved that binary Varshamov-Tenengolts codes correct one arbitrary edit error, i.e., one deletion \emph{or} one substitution, with nearly optimal redundancy.
However, this approach fails to extend to many simple and natural variations of the binary single-edit error setting.
In this work, we make progress on the code design problem above in three such variations:
\begin{itemize}
    \item We construct linear-time encodable and decodable length-$n$ non-binary codes correcting a single edit error with nearly optimal redundancy $\log n+O(\log\log n)$, providing an alternative simpler proof of a result by Cai, Chee, Gabrys, Kiah, and Nguyen (IEEE Trans.\ Inf.\ Theory 2021).
    This is achieved by employing what we call \emph{weighted VT sketches}, a notion that may be of independent interest.
    
    \item We construct linear-time encodable and list-decodable binary codes with list-size $2$ for one deletion \emph{and} one substitution with redundancy $4\log n+O(\log\log n)$. This matches the existential bound up to an $O(\log\log n)$ additive term.
    
    \item We show the existence of a binary code correcting one deletion \emph{or} one adjacent transposition with nearly optimal redundancy $\log n+O(\log\log n)$.
\end{itemize}
\end{abstract}

\section{Introduction}\label{sec:intro}

Deletions, substitutions, and transpositions are some of the most common types of errors jointly affecting information encoded in DNA-based data storage systems~\cite{YGM17,OAC+18}.
Therefore, it is natural to consider models capturing the interplay between these types of errors, along with the best possible codes for these settings.
More concretely, one usually seeks to pin down the optimal redundancy required to correct such errors, and also to design fast encoding and decoding procedures for low-redundancy codes.
It is well-known that deletions are challenging to handle even in isolation, since they cause a loss of synchronization between sender and receiver.
The situation where one aims to correct deletions in conjunction with other reasonable types of errors is even direr.
In fact, our understanding of this interplay remains scarce even in basic settings where only one or two such worst-case errors may occur.

One of the most fundamental settings where deletions interact with the other types of errors mentioned above is that of correcting a single \emph{edit} error (i.e., a deletion, insertion, or substitution) over a \emph{binary} alphabet.
In this case, linear-time encodable and decodable binary codes correcting a single edit error with nearly optimal redundancy have been known for more than 50 years.
Levenshtein~\cite{Lev65} showed that the binary Varshamov-Tenengolts (VT) code~\cite{VT65} defined as
\begin{equation}\label{eq:VTcode}
    \cC=\left\{x\in\bits^n:\sum_{i=1}^n i\cdot x_i = a \mod (2n+1)\right\}
\end{equation}
corrects one arbitrary edit error.
For appropriate choices of $a$ and $b$, this code has redundancy at most $\log n+2$, and it is not hard to see that at least $\log n$ bits of redundancy are required to correct one edit error.
Remarkably, a greedy Gilbert-Varshamov-type argument only guarantees the existence of single-edit correcting codes with redundancy $2\log n$ -- much higher than what can be achieved with the VT code.
We recommend Sloane's excellent survey~\cite{Slo02} for a more in-depth overview of binary VT codes and their connections to combinatorics.

Although the questions of determining the optimal redundancy and giving nearly-optimal explicit constructions of codes in the binary single-edit setting have been settled long ago, the underlying approach fails to extend to many simple, natural variations of this setting combining deletions with substitutions and transpositions.
In this work, we make progress on these questions in three such fundamental variations, which we proceed to describe next.

\subsection{\emph{Non-binary} single-edit correcting codes}

We begin by considering the problem of correcting a single arbitrary edit error over a non-binary alphabet.
This setting is especially relevant due to its connection to DNA-based data storage, which requires coding over a $4$-ary alphabet.
In this case, the standard \emph{VT sketch}
\begin{equation}\label{eq:VTsketch}
    f(x)=\sum_{i=1}^n i\cdot x_i\mod N,
\end{equation}
which allows us to correct one binary edit error in~\eqref{eq:VTcode} with an appropriate choice of $N$, is no longer enough.
Instead, we present a natural extension of the binary VT code to a non-binary alphabet via a new notion of \emph{weighted VT sketches}, which yields the following order-optimal result.
\begin{thm}\label{thm:nonbinedit}
There exists a $4$-ary\footnote{A $4$-ary alphabet is relevant for DNA-based data storage.} single-edit correcting code $\mathcal{C}\subseteq \{0,1,2,3\}^n$ with
\begin{equation*}
    \log n+\log\log n+7+o(1)
\end{equation*}
bits of redundancy, where $o(1)\to 0$ when $n\to\infty$. 
Moreover, there exists a single edit-correcting code $\mathcal{C}\subseteq \{0,1,2,3\}^n$ with $\log n+O(\log\log n)$ redundant bits that supports linear-time encoding and decoding.
The existential result extends to larger alphabet size $q$ with $\log n+O_q(\log\log n)$ redundant bits.
\end{thm}
This problem was previously considered by Cai, Chee, Gabrys, Kiah, and Nguyen~\cite{CCGKN21}, who proved an analogous result.
Our code construction requires slightly less redundancy and supports more efficient encoding and decoding procedures than the construction from~\cite{CCGKN21}.
However, we believe that our more significant contribution in this setting is the simpler approach we employ to prove Theorem~\ref{thm:nonbinedit} via weighted VT sketches. The technique of weighted VT sketches seems quite natural and powerful and may be of independent interest.
More details can be found in Section~\ref{sec:nonbinedit}, where we also present a more in-depth discussion on why the standard VT sketch~\eqref{eq:VTsketch} does not suffice in the non-binary case.

\subsection{Binary codes for one deletion \emph{and} one substitution}\label{sec:list2}

As our second contribution, we make progress on the study of single-deletion single-substitution correcting codes.
Recent work by Smagloy, Welter, Wachter-Zeh, and Yaakobi~\cite{SWWY20} constructed efficiently encodable and decodable binary single-deletion single-substitution correcting codes with redundancy close to $6\log n$.
On the other hand, it is known that $2\log n$ redundant bits are required, and a greedy approach shows the \emph{existence} of a single-deletion single-substitution correcting code with redundancy $4\log n+O(1)$.

In this setting, we ask what improvements are possible if we relax the unique decoding requirement slightly and instead require that the code be \emph{list-decodable with list-size $2$}.
In other words, our goal is to design a low-redundancy code $\cC\subseteq\bits^n$ such that for any corrupted string $y\in\bits^{n-1}\cup\bits^n$ there are at most two codewords $x,x'\in\cC$ that can be transformed into $y$ via some combination of at most one deletion and one substitution.
This is the strongest possible requirement after unique decoding, which corresponds to lists of size $1$.

The best known \emph{existential} upper bound on the optimal redundancy in the list-decoding setting is still $4\log n+O(1)$ via the Gilbert-Varshamov-type greedy algorithm.
In this paper, we give an explicit list-decodable code with list-size $2$ correcting one deletion and one substitution with redundancy matching the existential bound up to an $O(\log\log n)$ additive term.
At a high level, this code is obtained by combining the standard VT sketch~\eqref{eq:VTsketch} with \emph{run-based sketches}, which have been recently used in the design of two-deletion correcting codes~\cite{guruswami2021explicit}.
More precisely, we have the following result.
\begin{thm}\label{thm:delandsub}
There exists a linear-time encodable and decodable binary \emph{list-size $2$} single-deletion single-substitution correcting code $\cC\subseteq\bits^n$ with $4\log n+O(\log\log n)$ bits of redundancy.
\end{thm}

More details can be found in Section~\ref{sec:delandsub2}.

\subsection{Binary codes correcting one deletion \emph{or} one adjacent transposition}

Finally, we consider the interplay between deletions and adjacent transpositions, which map $01$ to $10$ and vice-versa.
An adjacent transposition may be seen as a special case of a burst of two substitutions.
Besides its relevance to DNA-based storage, the interplay between deletions and transpositions is an interesting follow-up to the single-edit setting discussed above because the VT sketch is highly ineffective when dealing with transpositions, while it is the staple technique for correcting deletions and substitutions.
The issue is that, if $y,y'\in\bits^n$ are obtained from $x\in\bits^n$ via any two adjacent transpositions of the form $01\mapsto 10$, then $f(y)=f(y')=f(x)-1$, where we recall $f(z)=\sum_{i=1}^n i\cdot z_i\mod N$ is the VT sketch.
This implies that knowing the VT sketch $f(x)$ reveals almost no information about the adjacent transposition, since correcting an adjacent transposition is equivalent to finding its location.

In this setting, the best known redundancy lower bound is $\log n$ (the same as for single-deletion correcting codes), while the the best known existential upper bound is $2\log n$, obtained by naively intersecting a single-deletion correcting code and a single-transposition correcting code.
A code with redundancy $\log n+O(1)$ for this setting was claimed in~\cite[Section III]{GYM18}, but the argument there is flawed.
In this work, we determine the optimal redundancy of codes in this setting up to an $O(\log\log n)$ additive term via a novel marker-based approach.
More precisely, we prove the following result.
\begin{thm}\label{thm:deltrans}
There exists a binary code $\cC\subseteq\bits^n$ correcting one deletion or one transposition with redundancy $\log n+O(\log\log n)$.
\end{thm}

In fact, since we know that every code that corrects one deletion also corrects one insertion~\cite{Lev65}, we can also conclude from Theorem~\ref{thm:deltrans} that there exists a binary code correcting one deletion, one insertion, or one transposition with nearly optimal redundancy $\log n+O(\log\log n)$.
More details can be found in Section~\ref{sec:deltrans}.

\subsection{Other related work}

Recently, there has been a flurry of works making progress in coding-theoretic questions analogous to the ones we consider here in other extensions of the binary single-edit error setting.

A line of work culminating in~\cite{BGZ18,guruswami2021explicit,SB21} has succeeded in constructing explicit low-redundancy codes correcting a constant number of worst-case deletions. Constructions focused on the two-deletion case have also been given, e.g., in \cite{SB21,GabrysS19,guruswami2021explicit}.
Explicit binary codes correcting a sublinear number of edit errors with redundancy optimal up to a constant factor have also been constructed recently~\cite{CJLW18,Hae19}.
Other works have considered the related setting where one wishes to correct a burst of deletions or insertions~\cite{SWGY17,LP20,WSF21}.
Following up on~\cite{SWWY20}, codes correcting a combination of more than one deletion and one substitution were given in~\cite{SPCH21} with sub-optimal redundancy.

List-decodable codes in settings with indel errors have also been considered before.
For example, Wachter-Zeh~\cite{Wac18} and Guruswami, Haeupler, and Shahrasbi~\cite{GHS21} study list-decodability from a linear fraction of deletions and insertions.
Most relevant to our result in Section~\ref{sec:list2}, Guruswami and H{\aa}stad~\cite{guruswami2021explicit} considered 
constructed an explicit list-size two code correcting two deletions with redundancy $3\log n$, thus beating the greedy existential bound in this setting.

The interplay between deletions and transpositions has also been considered before.
Gabrys, Yaakobi, and Milenkovic~\cite{GYM18} construct codes correcting a single deletion \emph{and} many adjacent transpositions.
In an incomparable regime, Schulman and Zuckerman~\cite{SZ99}, Cheng, Jin, Li, and Wu~\cite{CJLW19}, and Haeupler and Shahrasbi~\cite{HS18} construct explicit codes with good redundancy correcting a linear fraction of deletions and insertions and a nearly-linear fraction of transpositions.

\section{Preliminaries}

\subsection{Notation and conventions}

We denote sets by uppercase letters such as $S$ and $T$ or uppercase calligraphic letters such as $\cC$, and define $[n]=\{0,1,\dots,n-1\}$ and $S^{\leq k}=\bigcup_{i=0}^k S^i$ for any set $S$.
The symmetric difference between two sets $S$ and $T$ is denoted by $S\triangle T$.
We use the notation $\{\{a,a,b\}\}$ for multisets, which may contain several copies of each element.
Given two strings $x$ and $y$ over a common alphabet $\Sigma$, we denote their concatenation by $x\|y$ and write $x[i:j]=(x_i,x_{i+1},\dots,x_j)$.
We say $y\in\Sigma^k$ is a \emph{$k$-subsequence} of $x\in\Sigma^n$ if there are $k$ indices $1\leq i_1<i_2<\cdots<i_k\leq n$ such that $x_{i_j}=y_j$ for $j=1,\dots,k$, in which case we also call $x$ an \emph{$n$-supersequence} of $y$.
Moreover, we say $x[i:j]$ is an \emph{$a$-run} of $x$ if $x[i:j]=a^{j-i+1}$ for a symbol $a\in\Sigma$.
We denote the base-$2$ logarithm by $\log$.

A length-$n$ \emph{code} $\cC$ is a subset of $\Sigma^n$ for some alphabet $\Sigma$ which will be clear from context.
In this work, we are interested in the \emph{redundancy} of certain codes (measured in bits), which we define as
\begin{equation*}
    n\log |\Sigma|-\log|\cC|.
\end{equation*}

\subsection{Error models and codes}

Since we will be dealing with three distinct but related models of worst-case errors, we begin by defining the relevant standard concepts in a more general way.
We may define a worst-case error model over some alphabet $\Sigma$ by specifying a family of \emph{error balls}
\begin{equation*}
    \cB=\{B(y)\subseteq\Sigma^*:y\in\Sigma^*\}.
\end{equation*}
Intuitively, $B(y)$ contains all strings that can be corrupted into $y$ by applying an allowed error pattern.
We proceed to define unique decodability of a code $\cC\subseteq\Sigma^n$ with respect to an error model.
\begin{defn}[Uniquely decodable code]
We say a code $\cC\subseteq\Sigma^n$ is \emph{uniquely decodable (with respect to $\cB$)} if
\begin{equation*}
    |B(y)\cap \cC|\leq 1
\end{equation*}
for all $y\in\Sigma^*$.
\end{defn}
Throughout this work the underlying error model will always be clear from context, so we do not mention it explicitly.
We will also consider \emph{list-decodable} codes with small list size in Section~\ref{sec:delandsub2}, and so require the following more general definition.
\begin{defn}[List-size $t$ decodable code]
We say a code $\cC\subseteq\Sigma^n$ is \emph{list-size $t$ decodable (with respect to $\cB$)} if
\begin{equation*}
    |B(y)\cap \cC|\leq t
\end{equation*}
for all $y\in\Sigma^*$.
\end{defn}

Note that uniquely decodable codes correspond exactly to list-size-$1$ codes.
Moreover, we remark that for the error models considered in this work and constant $t$, the best existential for list-size-$t$ codes coincides with the best existential bound for uniquely decodable codes up to a constant additive term.

We proceed to describe the type of errors we consider.
A deletion transforms a string $x\in\Sigma^n$ into one of its $(n-1)$-subsequences.
An insertion transforms a string $x\in\Sigma^n$ into one of its $(n+1)$-supersequences.
A substitution transforms $x\in\Sigma^n$ into a string $x'\in\Sigma^n$ that differs from $x$ in exactly one coordinate.
An adjacent transposition transforms strings of the form $ab$ into $ba$. More formally, a string $x\in\Sigma^n$ is tranformed into a string $x'\in\Sigma^n$ with the property that $x'_k=x_{k+1}$ and $x'_{k+1}=x_k$ for some $k$, and $x'_i=x_i$ for $i\neq k,k+1$.

We can now instantiate the above general definitions under the specific error models considered in this paper.
In the case of a single edit, $B(y)$ contains all strings which can be transformed into $y$ via at most one deletion, one insertion, or one substitution.
In the case of one deletion \emph{and} one substitution, $B(y)$ contains all strings that can be transformed into $y$ by applying at most one deletion and at most one substitution.
Finally, in the case of one deletion or one adjacent transposition, $B(y)$ contains all strings that can be transformed into $y$ by applying either at most one deletion or at most one transposition.

\section{Non-binary single-edit correcting codes}\label{sec:nonbinedit}

In this section, we describe and analyze the code construction used to prove Theorem~\ref{thm:nonbinedit}.
Before we do so, we provide some intuition behind our approach.

\subsection{The binary alphabet case as a motivating example}\label{sec:bin}

It is instructive to start off with the binary alphabet case and the VT code described in~\eqref{eq:VTcode}, which motivates our approach for non-binary alphabets.
More concretely, we may wonder whether a direct generalization of $\cC$ to larger alphabets also corrects a single edit error, say
\begin{equation*}
    \cC'=\left\{x\in[q]^n \bigg| \sum_{i=1}^n i x_i=s\mod (1+2qn),\quad \forall c\in[q]:|\{i:x_i=c\}|=s_c\mod 2\right\},
\end{equation*}
where $[q]=\{0,1,\dots,q-1\}$.
However, this approach fails already over a ternary alphabet $\{0,1,2\}$.
In fact, $\cC'$ cannot correct worst-case deletions of $1$'s because it does not allow us to distinguish between
\begin{equation*}
    \dots \underline{1}02\dots \qquad \textrm{and}\qquad \dots 02\underline{1}\dots\;,
\end{equation*}
which can be obtained one from the other by deleting and inserting a $1$ in the underlined positions.
More generally, there exist codewords $x\in\cC'$ with substrings $(x_j=1,x_{j+1},\dots,x_k)$ not consisting solely of $1$'s satisfying
\begin{equation}\label{eq:badsub}
    \sum_{i=j+1}^k (x_i-1)=0.
\end{equation}
This implies that the string $x'$ obtained by deleting $x_j=1$ from $x$ and inserting a $1$ between $x_k$ and $x_{k+1}$ is also in $\cC'$.

In order to avoid the problem encountered by $\cC'$ above, we instead consider a \emph{weighted VT sketch} of the form
\begin{equation}\label{eq:weightedVT}
    f_w(x) = \sum_{i=1}^n i\cdot w(x_i) \mod N
\end{equation}
for some weight function $w:[q]\to\mathbb{Z}$ and an appropriate modulus $N$.
Using $f_w$ instead of the standard VT sketch $f(x)=\sum_{i=1}^n ix_i\mod N$ in the argument above causes the condition~\eqref{eq:badsub} for an uncorrectable $1$-deletion to be replaced by
\begin{equation*}
    \sum_{i=j+1}^k (w(x_{i})-w(1)) = 0.
\end{equation*}
Then, choosing $0\leq w(0)<w(1) < w(2)<\cdots<w(q-1)$ appropriately allows us to correct the deletion of a $1$ in $x$ given knowledge of $f_w(x)$ provided that $x$ satisfies a simple runlength constraint.
In turn, encoding an arbitrary message $z$ into a string $x$ satisfying this constraint can be done very efficiently via a direct application of the simple \emph{runlength replacement} technique from~\cite{SWGY17} using few redundant bits.
Theorem~\ref{thm:nonbinedit} is then obtained by instantiating the weighted VT sketch~\eqref{eq:weightedVT} with an appropriate weight function and modulus.

\subsection{Code construction}

In this section, we present our construction of a $4$-ary single-edit correcting code which leads to Theorem~\ref{thm:nonbinedit}.
As discussed in Section~\ref{sec:bin},
given an arbitrary string $x\in\{0,1,2,3\}^n$ we consider a weighted VT sketch
\begin{equation*}
    f(x)=\sum_{i=1}^n i\cdot w(x_i)\mod [1+2n\cdot(2\log n + 12)],
\end{equation*}
where $w(0)=0$, $w(1)=1$, $w(2)=2\log n + 11$, and $w(3)=2\log n + 12$,  along with the count sketches
\begin{equation*}
    h_c(x)= |\{i:x_i=c\}|\mod 2
\end{equation*}
for $c\in\{0,1,2\}$.
Intuitively, the count sketches allow us to cheaply narrow down exactly what type of deletion or substitution occurred (but not its position).
As we shall prove later on, successfully correcting the deletion of an $a$ boils down to ensuring that
\begin{equation}\label{eq:cond}
    \sum_{i=j}^k (w(x_i)-w(a))\neq 0
\end{equation}
for all $1\leq j\leq k\leq n$ such that there is $i\in[j,k]$ with $x_i\neq a$.
We call strings $x$ that satisfy this property for every $a$ \emph{regular}, and proceed to show that enforcing a simple runlength constraint on $x$ is sufficient to guarantee that it is regular.
\begin{lem}\label{lem:RLLregular4}
Suppose $x\in\{0,1,2,3\}^n$ satisfies the following property: \emph{If $x'$ denotes the subsequence of $x$ obtained by deleting all $1$'s and $3$'s and $x''$ denotes the subsequence obtained by deleting all $0$'s and $2$'s, it holds that all $0$-runs of $x'$ and all $3$-runs of $x''$ have length at most $\log n + 3$.}
Then, $x$ is regular.
\end{lem}
\begin{proof}
First, note that when $a=0,3$ it follows that \eqref{eq:cond} holds trivially for all $x$.
Thus, it suffices to consider $a=1,2$.
Fix any $x$ satisfying the property outlined in the lemma statement and $1\leq j\leq k\leq n$ such that there is $i\in[j,k]$ with $x_i\neq 1$.
The runlength constraint on $x'$ implies that there must be at least one $2$ in $x[j,k]$ for every consecutive subsequence of $2\lceil\log n+3\rceil$ $0$'s that appears in $x[j,k]$. 
Since $w(0)-w(1)=-1$ and $w(2)-1=2\log n+10>2\lceil\log n+3\rceil$, it follows that \eqref{eq:cond} holds.
The argument for the case $a=2$ is analogous using the fact that $w(3)-w(2)=1$ and $w(1)-w(2)=-(2\log n+11)<-2\lceil \log n+3\rceil$.
\end{proof}

Let $\mathcal{G}\subseteq\{0,1,2,3\}^n$ denote the set of regular strings. 
Given the above definitions, we set our code to be
\begin{equation}\label{eq:4arycode}
    \mathcal{C}=\mathcal{G}\cap \{x\in\{0,1,2,3\}^n:f(x)=s, h_c(x)=s_c, c\in\{0,1,2\}\}
\end{equation}
for appropriate choices of $s\in\{0,\dots,1+2n\cdot (2\log n+12)\}$ and $s_c\in\{0,1\}$ for $c=0,1,2$.
A straightforward application of the probabilistic method shows that most strings are regular.
\begin{lem}
Let $X$ be sampled uniformly at random from $\{0,1,2,3\}^n$. 
Then,
\begin{equation*}
    \Pr[\textnormal{$X$ is regular}]\geq 7/8.
\end{equation*}
\end{lem}
\begin{proof}
Let $X'$ and $X''$ be the subsequences of $X$ obtained by deleting all $1$'s and $3$'s or all $0$'s and $2$'s, respectively.
Then, the probability that $X'$ has a $0$-run of length $\log n+4$ starting at $i$ is $\frac{1}{16n}$.
By a union bound over the fewer than $n$ choices for $i$, it follows that $X'$ has at least one such $0$-run with probability at most $1/16$.
Since the same argument applies to $3$-runs in $X''$, a final union bound over the two events yields the desired result by Lemma~\ref{lem:RLLregular4}.
\end{proof}
As a result, by the pigeonhole principle there exist choices of $s,s_0,s_1,s_2$ such that
\begin{equation*}
    |\mathcal{C}|\geq \frac{7\cdot 4^n}{8\cdot 2^3\cdot (1+2n\cdot(2\log n+12))}.
\end{equation*}
This implies that we can make it so that $\mathcal{C}$ has $\log n+\log\log n+6+o(1)$ bits of redundancy, where $o(1)\to 0$ when $n\to\infty$, as desired.
If $n$ is not a power of two, then taking ceilings yields at most one extra bit of redundancy for a total of $\log n+\log\log n+7+o(1)$ bits, as claimed.

It remains to show that $\cC$ corrects a single edit in linear time and that a standard modification of $\cC$ admits a linear time encoder.
Observe that if a codeword $x\in\cC$ is corrupted into a string $y$ by a single edit error, we can tell whether it was a deletion, insertion, or substitution by computing $|y|$.
Therefore, we treat each such case separately below.

\subsection{Correcting one substitution}

Suppose that $y$ is obtained from some $x\in\mathcal{C}$ by changing an $a$ to a $b$ at position $i$.
Then, we can find $|w(a)-w(b)|$ by computing $h_a(y)-h_a(x)$ for $a=0,1,2$.
In particular, note that we can correctly detect whether no substitution was introduced, since this happens if and only if $h_a(y)=h_a(x)$ for $a=0,1,2$.
It also holds that
\begin{equation*}
    f(y)-f(x) = i\cdot(w(b)-w(a)).
\end{equation*}
Since $|i\cdot(w(b)-w(a))|\leq n\cdot (2\log n+12)< \frac{1+2n\cdot(2\log n+12)}{2}$, we can recover the position $i$ by computing
\begin{equation*}
    i = \frac{|f(y)-f(x)|}{|w(b)-w(a)|}.
\end{equation*}
Note that these steps can be implemented in time $O(n)$.

\subsection{Correcting one deletion}

Suppose that $y$ is obtained from $x\in\mathcal{C}$ by deleting an $a$ at position $i$.
First, note that we can find $a$ by computing $h_c(y)-h_c(x)$ for $c=0,1,2$.
Now, let $y^{(j)}$ denote the string obtained by inserting an $a$ to the left of $y_j$ (when $j=n$ this means we insert an $a$ at the end of $y$).
We have $x=y^{(i)}$ and our goal is to find $i$.
Consider $n\geq j\geq i$ and observe that
\begin{align*}
    f(x)-f(y^{(j)})&=f(y^{(i)})-f(y^{(j)})\\
    &=\sum_{\ell=i+1}^j (w(x_\ell)-w(a)),
\end{align*}
because $y_{\ell-1}=x_\ell$ for $\ell>i$.
Since $x$ is regular, it follows that $\sum_{\ell=i+1}^j (w(x_\ell)-w(a))\neq 0$ unless $x_{i+1}=\cdots=x_j=a$.
This suggests the following decoding algorithm: Successively compute $f(x)-f(y^{(j)})$ for $j=n,n-1,\dots, 1$ until  $f(x)-f(y^{(j)})=0$, in which case the above argument ensures that $y^{(j)}=x$ since we must be inserting $a$ into the same $a$-run of $x$ from which an $a$ was deleted.
This procedure runs in time $O(n)$.

\subsection{Correcting one insertion}

The procedure for correcting one insertion is very similar to that used to correct one deletion.\footnote{It is well known that every code that corrects one deletion also corrects one insertion~\cite{Lev65}. However, this implication does not hold in general if we require efficient decoding too.}
We present the argument for completeness.
Suppose $y$ is obtained from $x$ by inserting an $a$ between $x_{i-1}$ and $x_i$ (when $i=1$ or $i=n+1$ this means we insert an $a$ at the beginning or end of $x$, respectively).
First, observe that we can find $a$ by computing $h_c(x)-h_c(y)$ for $c=0,1,2$.
Let $y^{(j)}$ denote the string obtained from $y$ by deleting $y_j=a$.
Then, it holds that $y^{(i)}=x$ and for $j\geq i$ we have
\begin{align*}
    f(x)-f(y^{(j)})&=f(y^{(i)})-f(y^{(j)})\\
    &= -\sum_{\ell=i}^{j-2} (w(x_\ell)-w(a)),
\end{align*}
because $y_\ell=x_{\ell-1}$ when $j>i$.
As before, using the fact that $x$ is regular allows us to conclude that $f(x)-f(y^{(j)})=0$ if and only if $x_i=\cdots=x_{j-2}=a$, in which case we are deleting an $a$ from the correct $a$-run of $x$.
Therefore, we can correct an insertion of an $a$ in $x$ by successively computing $f(x)-f(y^{(j)})$ for all $j$ such that $y_j=a$ starting at $j={n+1}$ and deleting $y_j$ for the first $j$ such that $f(x)-f(y^{(j)})=0$, in which case the argument above ensures that $y^{(j)}=x$.
This procedure runs in time $O(n)$.

\subsection{A linear-time encoder}\label{sec:linenc}

In the previous sections we described a linear-time decoder that corrects a single edit error in regular strings $x$ assuming knowledge of the weighted VT sketch $f(x)$ and the count sketches $h_c(x)$ for $c=0,1,2$.
It remains to describe a low-redundancy linear-time encoding procedure for a slightly modified version of our code $\cC$ defined in~\eqref{eq:4arycode}.
Fix an arbitrary message $z\in\{0,1,2,3\}^m$.
We proceed in two steps:
\begin{enumerate}
    \item We describe a simple linear-time procedure based on runlength replacement that encodes $z$ into a regular string $x\in\{0,1,2,3\}^{m+4}$;
    
    \item We append an appropriate encoding of the sketches $f(x)\|h_0(x)\|h_1(x)\|h_2(x)$ (which we now see as binary strings) to $x$ that can be recovered even if the final string is corrupted by an edit error. This adds $O(\log\log n)$ bits of redundancy.
\end{enumerate}

We begin by considering the first step.
We can encode $z$ into a regular string $x\in\{0,1,2,3\}^{m+4}$ by enforcing a runlength constraint using a simple runlength replacement technique~\cite[Appendix B]{SWGY17}.
\begin{lem}\label{lem:linenc}
There is a linear-time procedure $\Enc$ that given $z\in\{0,1,2,3\}^m$ outputs $x=\Enc(z)\in\{0,1,2,3\}^{m+4}$ with the following property: If $x'$ is obtained by deleting all $1$'s and $3$'s from $x$ and $x''$ is obtained by deleting all $0$'s and $2$'s, it holds that all $0$-runs of $x'$ and all $3$-runs of $x''$ have length at most $\lceil\log n+2\rceil$.
Moreover, there is a linear-time procedure $\Dec$ such that $\Dec(x)=z$.
In particular, $x$ is regular by Lemma~\ref{lem:RLLregular4}.
\end{lem}
\begin{proof}
Let $z'\in\{0,2\}^{m'}$ with $m'\leq m$ denote the subsequence at positions $1\leq i_1<\cdots<i_{m'}\leq m$ of $z$ obtained by deleting all $1$'s and $3$'s, and let $z''\in\{0,1,2,3\}^{m-m'}$ denote the leftover subsequence.
We may apply the runlength replacement technique from~\cite{SWGY17} to $z'$ and $z''$ separately in order to obtain strings $x'$ and $x''$ with the desired properties.
For completeness, we describe it below.
The final encoding $x$ is obtained by inserting the symbols of $x'$ into the positions $i_1,\dots,i_{m'},m+1,m+2$ of $z$ and the symbols of $x''$ into the remaining positions.

The encoding of $z'$ into $x'$ proceeds as follows:
First, append the string $20$ to $z'$.
Then, scan $z'$ from left to right.
If a $0$-run of length $\lceil\log m' + 2\rceil$ is found starting at $i$, then we remove it from $z'$ and append the marker $\mathsf{bin}(i)\|22$ to $z'$, where $\mathsf{bin}(i)$ denotes the binary expansion of $i$ (over $\{0,2\}$ instead of $\{0,1\}$) to $\lceil \log m'\rceil$ bits.
Note that the length of $z'$ stays the same after each such operation, and the addition of a marker does not introduce new $0$-runs of length $\lceil \log m'+2\rceil$.
Repeating this procedure until no more $0$-runs of length $\lceil \log m'+2\rceil$ are found yields a string $x'\in\{0,2\}^{m'+2}$ without $0$-runs of length $\lceil \log m'+2\rceil\leq \lceil \log m+2\rceil$.
This procedure, along with the transformation from $x'$ to $x$, runs in time $O(n)$.
The encoding of $z''$ into $x''\in\{1,3\}^{m-m'+2}$ is analogous with $1$ in place of $2$ and $3$ in place of $0$.

It remains to describe how to recover $z$ from $x$.
It suffices to describe how to recover $z'$ and $z''$ from $x'$ and $x''$, respectively, in time $O(n)$.
By the encoding procedure above, we know that if $x'$ ends in a $0$ then it follows that $z'=x'[1:m'=|x'|-2]$.
If $x'$ ends in a $2$, it means that $x'$ has suffix $\mathsf{bin}(i)\|22$ for some $i$.
Then, we recover $i$ from this suffix and insert a $0$-run of length $\lceil \log m' +2\rceil$ in the appropriate position of $x'$.
We repeat this until $x'$ ends in a $0$.
This procedure also runs in time $O(n)$.
The approach for $x''$ is analogous and yields $z''$.
Finally, we can merge $z'$ and $z''$ correctly to obtain $z$ since we know that $z'$ should be inserted into the positions occupied by $x'$ in $x$ (disregarding the last two symbols of $x'$).
\end{proof}

To finalize the description of the overall encoding procedure, let $x=\Enc(z)$ and define $(\overline{\Enc},\overline{\Dec})$ to be an explicit coding scheme for strings of length $\ell=|f(x)\|h_0(x)\|h_1(x)\|h_2(x)|$ correcting a single edit error (a naive construction has redundancy $2\log\ell+O(1)=O(\log\log m)$).
If
\begin{equation*}
    u=\overline{\Enc}(f(x)\|h_0(x)\|h_1(x)\|h_2(x)),
\end{equation*}
the final encoding procedure is
\begin{equation*}
    z \mapsto x\|u\in\{0,1,2,3\}^n,
\end{equation*}
which runs in time $O(m)=O(n)$ and has overall redundancy $\log m+O(\log\log m)=\log n+O(\log\log n)$.

Now, suppose $y$ is obtained from $x\|u$ by introducing one edit error.
We show how to recover $z$ in time $O(n)$ from $y$.
First, we can recover $u$ by running $\overline{\Dec}$ on the last $|u|-1$, $|u|$, or $|u|+1$ symbols of $y$ depending on whether a deletion, substitution, or insertion occurred, respectively.
If the last $|u|$ symbols of $y$ are not equal to $u$, we know that the edit error occurred in that part of $y$. Therefore, we have $x=y[1:m+4]$ and can compute $z=\Dec(y[1:m+4])$.
Else, if the last $|u|$ symbols of $y$ are equal to $u$, it follows that $y[1:|y|-|u|]$ can be obtained from $x$ via one edit error. This means we can recover $x$ from $y[1:|y|-|u|]$ and in turn compute $z=\Dec(x)$.

\section{Binary list-size two code for one deletion and one substitution}\label{sec:delandsub2}

In this section, we describe and analyze a binary list-size two decodable code for one deletion and one substitution, which yields Theorem~\ref{thm:delandsub}.
Departing from the approach of~\cite{SWWY20}, our construction makes use of \emph{run-based sketches} combined with the standard VT sketch.
Run-based sketches have thus far been exploited in the construction of multiple-deletion correcting codes, including list-decodable codes with small list size~\cite{guruswami2021explicit}.

We proceed to describe the required concepts:
Given a string $x=(x_1,\dots,x_n)\in\bits^n$, we define its \emph{run string} $r^x$ by first setting $r^x_0=0$ along with $x_0=0$ and $x_{n+1}=1$, and then iteratively computing $r^x_{i}=r^x_{i-1}$ if $x_{i}=x_{i-1}$ and $r^x_{i}=r^x_{i-1}+1$ otherwise for $i=1,\dots,n,n+1$.
Note that every string $x$ is uniquely determined by its run string $r^x$ and vice-versa.
Moreover, it holds that $r^x$ defines a non-decreasing sequence and $0\leq r^x_i\leq i$ for every $i=1,\dots,n,n+1$.
As an example, the run string corresponding to $x = 011101000$ is $r^x = 0111234445$.
We call $r^x_i$ the \emph{rank} of index $i$ in $x$.
We will denote the total number of runs in $x$ by $r(x)$.
The following simple structural lemma about the number of runs in a corrupted string will prove useful in our case analysis. 
\begin{lem}
    If $x'$ is obtained from $x$ via one deletion, then either $r(x')=r(x)$ or $r(x')=r(x)-2$.
    On the other hand, if $x'$ is obtained from $x$ via one substitution, then either $r(x')=r(x)$, $r(x')=r(x)-2$, or $r(x')=r(x)+2$.
\end{lem}
\begin{proof}
The desired statement follows by case analysis.
We have $r(x')=r(x)-2$ when $x'$ is obtained by deleting or flipping a bit in a run of length $1$ in $x$.
Otherwise, we have $r(x')=r(x)$ when $x'$ is obtained by deleting a bit in a run of length at least $2$ or by flipping the leftmost or rightmost bit in a run of length at least $2$.
In the remaining case where the flipped bit is in the middle of a run of length at least $3$, we have $r(x')=r(x)+2$.
\end{proof}

The main component of our code is a combination of the standard VT sketch
\begin{equation}\label{eq:vt1del1sub}
    f(x)=\sum_{i=1}^n ix_i\mod (3n+1)
\end{equation}
with the run-based sketches
\begin{align}
    &f_1^r(x)=\sum_{i=1}^{n} r^x_i\mod 12 n+1,\label{eq:run1}\\
    &f_2^r(x)=\sum_{i=1}^{n} r^x_i(r^x_i-1) \mod 16n^2+1\label{eq:run2}
\end{align}
originally considered in~\cite{guruswami2021explicit}.
Additionally, we also consider the count sketches
\begin{align}
    &h(x)= \sum_{i=1}^n x_i \mod 5,\label{eq:count1delsub}\\
    &h_r(x) = r(x)\mod 13.\label{eq:count2delsub}
\end{align}
The count sketches are used to distinguish different error patterns.
Intuitively, the sketch $h(x)$ is used to determine the value of the bit deleted and the value of the bit flipped, while $h_r(x)$ is used to distinguish among different cases of run changes, i.e., whether the number of runs decreases by four or increases by two due to the errors. 
For each possible error pattern, we use the standard VT-sketch and the run-based sketches to decode.
Given the above, our code is defined to be
\begin{equation}\label{eq:codedelsub}
    \mathcal{C} = \{x\in\{0,1\}^n: f(x) = s, f_1^r(x) = s_1^r, f_2^r(x) = s_2^r, h(x) = u, h_r(x) = u_r\},
\end{equation}
for an appropriate choice of $s\in[3n+1]$, $s_1^r\in[12n+1]$, $s_2^r\in[16n^2+1]$, $u\in[5]$, and $u_r\in[13]$.
By the pigeonhole principle, there is such a choice which ensures $\cC$ has redundancy $4\log n+O(1)$.
In the remainder of this section, we will show that $\cC$ admits linear-time list-size two decoding from one deletion and one substitution.
Moreover, we will also show that a slightly modified version of $\cC$ with redundancy $4\log n+O(\log\log n)$ admits linear-time encoding and inherits the error correction properties of $\cC$.

\subsection{Error correction properties}\label{sec:correct1del1sub}
Let $y$ be the string obtained from $x\in\cC$ after one deletion at index $d$ and one substitution at index $e$. 
We use $x_e$ to denote the bit flipped, and $x_d$ to denote the bit deleted in $x$.
When $d = e$, we simply have one deletion and no substitution. 
Our goal is to recover $x$ from $y$.

For our analysis, it is useful to note that
\begin{align*}
    &- 4n \leq \sum_{i=1}^n r^y_i - \sum_{i=1}^n r^x_i \leq  2n,\\
    &  - 5n^2 \leq \sum_{i=1}^n r^y_i(r^y_i-1) - \sum_{i=1}^n r^x_i(r^x_i-1) \leq + 3n^2.
\end{align*}
In particular, by the choice of modulus in the run-based sketches $f^r_1$ and $f^r_2$, this implies that we can recover the values $\sum_{i=1}^n r^x_i$ and $\sum_{i=1}^n r^x_i(r^x_i-1)$ from knowledge of $y$ and the sketches $f^r_1(x)$ and $f^r_2(x)$.
Therefore, we can safely omit the modulus when we compare the difference between the sketches of $x$ and of $y$ later in our analysis. 

Moreover, it is useful to observe that one deletion and no substitution can be equivalently transformed to one deletion and one substitution, thus we will only consider the case in which we have one deletion and one substitution, i.e., $d \neq e$.
In fact:
\begin{itemize}
    \item If the single deletion of $b$ in the $i$-th run does not change the number of runs, then it is equivalent to one deletion of $1-b$ in the $(i-1)$-th run and one substitution at the beginning of the $i$-th run. See Figure~\ref{fig:single_del_1} for an example.
    
    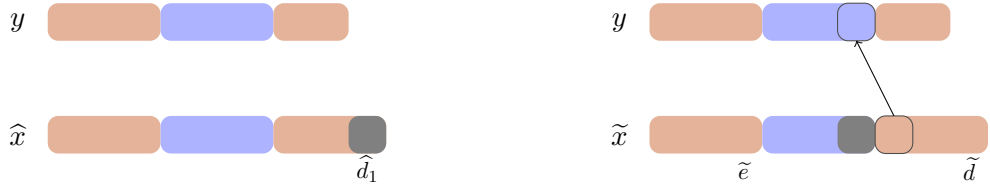
\begin{figure}[h!]        
    \centering
    \begin{tikzpicture}
    \draw (-2.4,0.24) node {$y$};
    \draw (2.25,-1.7) node [scale = 0.8] {$\widehat{d}_1$};
 
    \fill[myred!50, rounded corners] (-2, -1.5) rectangle (-0.5, -1);
    \fill[myblue!50, rounded corners] (-0.5, -1.5) rectangle (1, -1);
    \fill[myred!50, rounded corners] (1, -1.5) rectangle (2.5, -1);
    
    \fill[black!50, rounded corners] (2, -1.5) rectangle (2.5, -1);

    \draw (-2.4,-1.25) node {$\widehat{x}$};
    \fill[myred!50, rounded corners] (-2, 0) rectangle (-0.5, 0.5);
    \fill[myblue!50, rounded corners] (-0.5, 0) rectangle (1, 0.5);
    \fill[myred!50, rounded corners] (1, 0) rectangle (2, 0.5);
    
    \draw (5.6,0.24) node {$y$};
    \draw (10.25,-1.7) node [scale = 0.8] {$\widetilde{d}$};
    \draw (7.25 ,-1.73) node [scale = 0.8]{$\widetilde{e}$};
    \fill[myred!50, rounded corners] (6, -1.5) rectangle (7.5, -1);
    \fill[myblue!50, rounded corners] (7.5, -1.5) rectangle (9, -1);
    \fill[myred!50, rounded corners] (9, -1.5) rectangle (10.5, -1);
    
    \fill[black!50, rounded corners] (8.5, -1.5) rectangle (9, -1);
    \draw[black!70, rounded corners] (9, -1.5) rectangle (9.5, -1);
    \draw (5.6,-1.25) node {$\widetilde{x}$};
    \fill[myred!50, rounded corners] (6, 0) rectangle (7.5, 0.5);
    \fill[myblue!50, rounded corners] (7.5, 0) rectangle (9, 0.5);
    \fill[myred!50, rounded corners] (9, 0) rectangle (10, 0.5);
    \draw[black!70, rounded corners] (8.5, 0) rectangle (9, 0.5);
    \draw[->] (9.25,-1) -- (8.75,0);
    \end{tikzpicture}
    \caption{Transforming a single deletion into one deletion and one substitution, when the single deletion does not change the number of runs.}
    \label{fig:single_del_1}
\end{figure}

    \item If the single deletion of $b$ in the $i$-th run reduces the number of runs by two, then it is equivalent to one deletion in the $(i-1)$-th run and a substitution in the $i$-th run. See Figure~\ref{fig:single_del_2} for an example.
    
    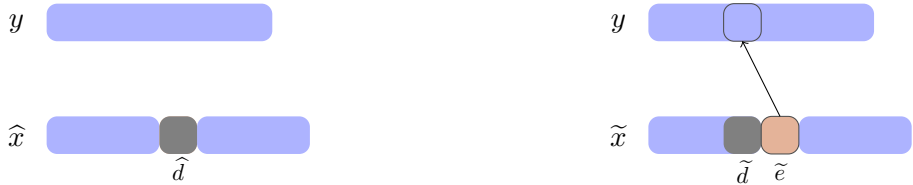
\begin{figure}[h!]
    \centering
    \begin{tikzpicture}
    \draw (-2.4,0.24) node {$y$};
    \draw (-0.25,-1.7) node [scale = 0.8] {$\widehat{d}$};
    \fill[myblue!50, rounded corners] (-2, -1.5) rectangle (-0.5, -1);
     \fill[myred!50, rounded corners] (-0.5, -1.5) rectangle (0, -1);
     \fill[myblue!50, rounded corners] (0, -1.5) rectangle (1.5, -1);
    
    \fill[black!50, rounded corners] (-0.5, -1.5) rectangle (0, -1);
    
    \draw (-2.4,-1.25) node {$\widehat{x}$};
    \fill[myblue!50, rounded corners] (-2, 0) rectangle (1, 0.5);
    
    \draw (5.6,0.24) node {$y$};
    \draw (7.25 ,-1.73) node [scale = 0.8]{$\widetilde{d}$};
    \draw (7.75 ,-1.73) node [scale = 0.8]{$\widetilde{e}$};
    \fill[myblue!50, rounded corners] (6, -1.5) rectangle (7.5, -1);
    \fill[myred!50, rounded corners] (7.5, -1.5) rectangle (8, -1);
    \fill[myblue!50, rounded corners] (8, -1.5) rectangle (9.5, -1);
    
    \fill[black!50, rounded corners] (7, -1.5) rectangle (7.5, -1);
    \draw (5.6,-1.25) node {$\widetilde{x}$};
    \fill[myblue!50, rounded corners] (6, 0) rectangle (9, 0.5);
    \draw[black!70, rounded corners] (7.5, -1.5) rectangle (8, -1);
    \draw[black!70, rounded corners] (7, 0) rectangle (7.5, 0.5);
    \draw[->] (7.75,-1) -- (7.25,0);
    
    \end{tikzpicture}
        \caption{Transforming a single deletion into one deletion and one substitution, when the single deletion reduces two runs.}
        \label{fig:single_del_2}
\end{figure}
\end{itemize}

Let $\delta$ be an indicator variable of whether $e > d$. That is, $\delta = 1$ if $e > d$, and $\delta=0$ otherwise. Then, the corrupted string $y$ can be regarded as a string obtained via one substitution at index $e-\delta$ from $x'\in\{0,1\}^{n-1}$, where $x'$ is in turn obtained via one deletion from $x$ at index $d$. 
See Figure~\ref{fig:example_delta} for an example of $\delta = 0$ or $\delta = 1$.

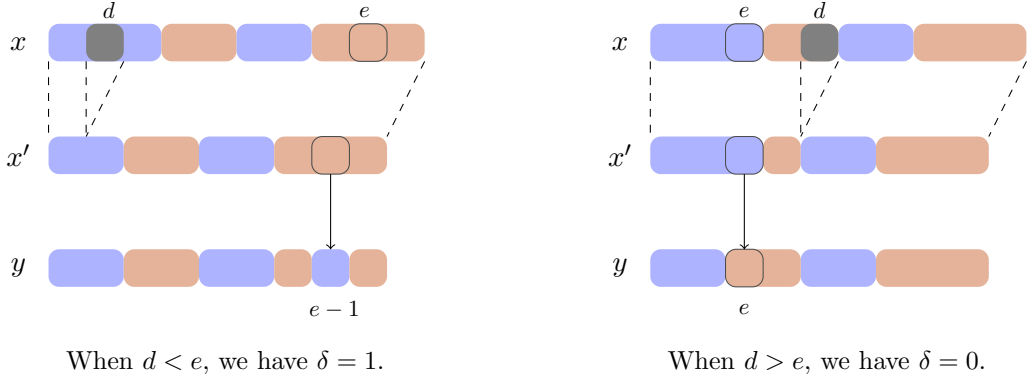
\begin{figure}[ht]
    \centering
    \begin{tikzpicture}
    \draw (-2.4,0.24) node {$x$};
    \draw (-1.2,0.69) node [scale = 0.8] {$d$};
    \draw (2.2,0.66) node [scale = 0.8]{$e$};
    \fill[myblue!50, rounded corners] (-2, 0) rectangle (-0.5, 0.5);
    \fill[myred!50, rounded corners] (-0.5, 0) rectangle (0.5, 0.5);
    \fill[myblue!50, rounded corners] (0.5, 0) rectangle (1.5, 0.5);
    \fill[myred!50, rounded corners] (1.5, 0) rectangle (3, 0.5);
    
    \fill[black!50, rounded corners] (-1.5, 0) rectangle (-1, 0.5);
    \draw[black!70, rounded corners] (2, 0) rectangle (2.5, 0.5);
    \draw (-2.4,-1.25) node {$x'$};
    \fill[myblue!50, rounded corners] (-2, -1.5) rectangle (-1, -1);
    \fill[myred!50, rounded corners] (-1, -1.5) rectangle (0, -1);
    \fill[myblue!50, rounded corners] (0, -1.5) rectangle (1, -1);
    \fill[myred!50, rounded corners] (1, -1.5) rectangle (2.5, -1);
    
    \draw[dashed] (-2,0) -- (-2,-1);
    \draw[dashed] (-1.5,0) -- (-1.5,-1);
    \draw[dashed] (-1,0) -- (-1.5,-1);
    \draw[dashed] (3,0) -- (2.5,-1);
    \draw[black!70, rounded corners] (1.5, -1.5) rectangle (2, -1);
    \draw (-2.4,-2.75) node {$y$};
    \fill[myblue!50, rounded corners] (-2, -3) rectangle (-1, -2.5);
    \fill[myred!50, rounded corners] (-1, -3) rectangle (0, -2.5);
    \fill[myblue!50, rounded corners] (0, -3) rectangle (1, -2.5);
    \fill[myred!50, rounded corners] (1, -3) rectangle (1.5, -2.5);
    \fill[myblue!50, rounded corners] (1.5, -3) rectangle (2, -2.5);
    \fill[myred!50, rounded corners] (2, -3) rectangle (2.5, -2.5);
    \draw[->] (1.75,-1.5) -- (1.75,-2.5);
    \draw (1.8,-3.3) node [scale = 0.8] {$e-1$};
    \draw (0.35,-4) node [scale = 0.9] {When $d<e$, we have $\delta = 1$.};
    
    \draw (5.6,0.24) node {$x$};
    \draw (8.25,0.69) node [scale = 0.8] {$d$};
    \draw (7.25 ,0.66) node [scale = 0.8]{$e$};
    \fill[myblue!50, rounded corners] (6, 0) rectangle (7.5, 0.5);
    \fill[myred!50, rounded corners] (7.5, 0) rectangle (8.5, 0.5);
    \fill[myblue!50, rounded corners] (8.5, 0) rectangle (9.5, 0.5);
    \fill[myred!50, rounded corners] (9.5, 0) rectangle (11, 0.5);
    
    \fill[black!50, rounded corners] (8, 0) rectangle (8.5, 0.5);
    \draw[black!70, rounded corners] (7, 0) rectangle (7.5, 0.5);
    \draw (5.6,-1.25) node {$x'$};
    \fill[myblue!50, rounded corners] (6, -1.5) rectangle (7.5, -1);
    \fill[myred!50, rounded corners] (7.5, -1.5) rectangle (8, -1);
    \fill[myblue!50, rounded corners] (8, -1.5) rectangle (9, -1);
    \fill[myred!50, rounded corners] (9, -1.5) rectangle (10.5, -1);
    
    \draw[dashed] (6,0) -- (6,-1);
    \draw[dashed] (8,0) -- (8,-1);
    \draw[dashed] (8.5,0) -- (8,-1);
    \draw[dashed] (11,0) -- (10.5,-1);
    \draw[black!70, rounded corners] (7, -1.5) rectangle (7.5, -1);
    
    \draw (5.6,-2.75) node {$y$};
    \fill[myblue!50, rounded corners] (6, -3) rectangle (7, -2.5);
    \fill[myblue!50, rounded corners] (8, -3) rectangle (9, -2.5);
    \fill[myred!50, rounded corners] (9, -3) rectangle (10.5, -2.5);
    \fill[myred!50, rounded corners] (7, -3) rectangle (8, -2.5);
    \draw[black!70, rounded corners] (7, -3) rectangle (7.5, -2.5);
    \draw[->] (7.25,-1.5) -- (7.25,-2.5);
    \draw (7.25,-3.3) node [scale = 0.8] {$e$};
    \draw (8.35,-4) node [scale = 0.9] {When $d>e$, we have $\delta = 0$.};
    \end{tikzpicture}
    \caption{Example of $d>e$ and $d<e$. The blue rectangles denote runs of 0, and the red rectangles denote runs of 1. 
    The shaded rectangle denotes the deleted bit and the bordered rectangle denotes the flipped bit.} 
    \label{fig:example_delta}
\end{figure}

The process of decoding can be thought of as inserting a bit $x_d$
before the $d$-th bit in $y$ and flipping the $(e-\delta)$-th bit in $y$. 
Roughly speaking, we will begin with a
candidate position pair $(\widetilde{d},\widetilde{e})$ with $\widetilde{d}$ is as small as possible with the property that, if $\widetilde{x}$ denotes the string obtained from $y$ by inserting $x_d$ before the $d$-th bit in $y$ and flipping the bit at position $\widetilde{e}-\widetilde{\delta}$ in $y$, where $\widetilde{\delta}$ indicates whether $\widetilde{d}<\widetilde{e}$,
then $f(\widetilde{x})=f(x)$.
Then we move $\widetilde{d}$ to the right to find the next candidate pair $\widetilde{d}$ such that $\widetilde{d}$ is as small as possible, and $\widetilde{e}$ is the unique index such that $(\widetilde{d},\widetilde{e})$ gives the correct value of $f(x)$.
During this process, either $f_1^r(x)$ determines a unique candidate position pair $(\widetilde{d},\widetilde{e})$, or a convexity-type property of $f_2^r(x)$ guarantees at most two candidate position pairs.
The convexity of $f_2^r(x)$ is a consequence of the following lemma.
\begin{lem}[\protect{\cite[Lemma 4.1]{guruswami2021explicit}}]
\label{lem:f2_convex}
Let $a_i$ and $a'_i$ be two sequences of non-negative integers such that $\sum_{i=1}^n {a_i} = \sum_{i=1}^n {a'_i}$ and there is a value $t$ such that for all $i$ satisfying $a_i < a'_i$ it holds that $a'_i \leq t$, and for all $i$ satisfying $a_i > a'_i$ it holds that $a'_i \geq t$. Then, either $a_i=a'_i$ for all $i$, or
\begin{equation*}
    \sum_{i=1}^n a_i(a_i - 1) > \sum_{i=1}^n a'_i (a'_i - 1).
\end{equation*}
\end{lem}

Given $h_r(x)$ and $h_r(y)$, we can easily tell how the number of runs changes; with the sketches $h(x)$ and $h(y)$, we can recover the values of $x_d$ and $x_e$.
Table~\ref{tab:h(x)} summarizes the connections between $h$ and the values of the corrupted bits.

\begin{table}[h!]
    \centering
    \begin{tabular}{|c|c|c|}
        \hline
        Difference between $h(x)$ and $h(y)$ & Bit deleted & Bit flipped\\
        \hline
         $h(x) - h(y) = -1$ & $x_d = 0$ & $x_e = 0$ \\
         \hline
         $h(x) - h(y) = \ \ 1$ & $x_d = 0$ & $x_e = 1$\\
         \hline
         $h(x) - h(y) = \ \  0$ & $x_d = 1$ & $x_e = 0$\\
         \hline
         $h(x) - h(y) = \ \ 2$ & $x_d = 1$ & $x_e = 1$\\
         \hline
    \end{tabular}
    \caption{Correspondence between $(h(x),h(y))$ and $(x_d, x_e)$.}
    \label{tab:h(x)}
\end{table}

\paragraph{Elementary moves.} 
Recall that we defined the string $\widetilde{x}$ associated with a candidate position pair $(\widetilde{d},\widetilde{e})$ to be the string obtained from $y$ by inserting $x_d$ before $y_{\widetilde{d}}$ and flipping $y_{\widetilde{e}-\widetilde{\delta}}$.
Making a parallel with $x$ and $x'$, let $\widetilde{x}'$ denote the string obtained from $\widetilde{x}$ by deleting $\widetilde{x}_{\widetilde{d}}$.
Then, we say $(\widetilde{d},\widetilde{e})$ is a \emph{valid} pair if $f(\widetilde{x})=f(x)$, $h_r(\widetilde{x})=h_r(x)$, and $h_r(x')=h_r(\widetilde{x}')$.
Intuitively, valid pairs are indistinguishable from the true error pattern $(d,e)$ from the perspective of the VT-sketch and the count sketches.

During the decoding procedure, we move from one valid pair to another as follows: 
Suppose we hold the valid pair $(\widetilde{d},\widetilde{e})$.
Then, we move $\widetilde{d}$ one index to the right, and check whether the unique $\widetilde{e}$ such that $f(\widetilde{x}) = f(x)$ forms a valid pair with $\widetilde{d}$. 
If not, we again move $\widetilde{d}$ one index to the right and repeat the process until we find the next valid pair.
We call this an \emph{elementary move}.
Note that since inserting a bit $b$ into a $b$-run at any position gives the same output, we may always move $\widetilde{d}$ to the end of the next $x_d$-run in $y$ (which may be empty).

For example, suppose that the error pattern indicates that $x_d = 1$, $x_e =1$, the deletion does not reduce the number of runs, i.e., $h_r(x') = h_r(x)$, while the substitution increases the number of runs by two, i.e., $h_r(x') - h_r(y) = -2 $. 
This means that in an elementary move, $\widetilde{d}$ is moving from the end of a $1$-run in $y$ to the end of the next $1$-run in $y$, while $\widetilde{e}$ is moving to the left accordingly so that the pair is valid. 
Moreover, since this elementary move needs to match the error pattern that the substitution increases two runs, it must satisfy that $y_{\widetilde{e}-\widetilde{\delta}+1} = y_{\widetilde{e}-\widetilde{\delta}-1} = 0$ both before and after the move. Figure~\ref{fig:elementary} shows an example of such an elementary move.

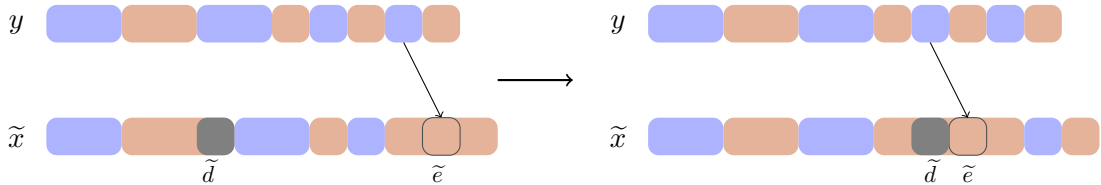
\begin{figure}[h!]
    \centering
    \begin{tikzpicture}
    \draw (0.15,-1.7) node [scale = 0.8] {$\widetilde{d}$};
    \draw (3.2,-1.75) node [scale = 0.8]{$\widetilde{e}$};
    \draw (-2.4,-1.25) node {$\widetilde{x}$};
    \fill[myblue!50, rounded corners] (-2, 0) rectangle (-1, 0.5);
    \fill[myred!50, rounded corners] (-1, 0) rectangle (0, 0.5);
    \fill[myblue!50, rounded corners] (0, 0) rectangle (1, 0.5);
    \fill[myred!50, rounded corners] (1, 0) rectangle (1.5, 0.5);
    \fill[myblue!50, rounded corners] (1.5, 0) rectangle (2, 0.5);
    \fill[myred!50, rounded corners] (2, 0) rectangle (2.5, 0.5);
    \fill[myblue!50, rounded corners] (2.5, 0) rectangle (3, 0.5);
    \fill[myred!50, rounded corners] (3, 0) rectangle (3.5, 0.5);
    \draw[->] (2.75,0) -- (3.25,-1);
    
    \draw (-2.4,0.24) node {$y$};
    \fill[myblue!50, rounded corners] (-2, -1.5) rectangle (-1, -1);
    \fill[myred!50, rounded corners] (-1, -1.5) rectangle (0.5, -1);
    \fill[myblue!50, rounded corners] (0.5, -1.5) rectangle (1.5, -1);
    \fill[myred!50, rounded corners] (1.5, -1.5) rectangle (2, -1);
    \fill[myblue!50, rounded corners] (2, -1.5) rectangle (2.5, -1);
    \fill[myred!50, rounded corners] (2.5, -1.5) rectangle (4, -1);
    \draw[black!70, rounded corners] (3, -1.5) rectangle (3.5, -1);
    \fill[black!50, rounded corners] (0, -1.5) rectangle (0.5, -1);
    
    \draw (5.6,0.24) node {$y$};
    \draw (9.75,-1.7) node [scale = 0.8] {$\widetilde{d}$};
    \draw (10.25 ,-1.75) node [scale = 0.8]{$\widetilde{e}$};
    \fill[myblue!50, rounded corners] (6, 0) rectangle (7, 0.5);
    \fill[myred!50, rounded corners] (7, 0) rectangle (8, 0.5);
    \fill[myblue!50, rounded corners] (8, 0) rectangle (9, 0.5);
    \fill[myred!50, rounded corners] (9, 0) rectangle (9.5, 0.5);
    \fill[myblue!50, rounded corners] (9.5, 0) rectangle (10, 0.5);
    \fill[myred!50, rounded corners] (10, 0) rectangle (10.5, 0.5);
    \fill[myblue!50, rounded corners] (10.5, 0) rectangle (11, 0.5);
    \fill[myred!50, rounded corners] (11, 0) rectangle (11.5, 0.5);
    
    \draw (5.6,-1.25) node {$\widetilde{x}$};
    \fill[myblue!50, rounded corners] (6, -1.5) rectangle (7, -1);
    \fill[myred!50, rounded corners] (7, -1.5) rectangle (8, -1);
    \fill[myblue!50, rounded corners] (8, -1.5) rectangle (9, -1);
    \fill[myred!50, rounded corners] (9, -1.5) rectangle (11, -1);
    \fill[myblue!50, rounded corners] (11, -1.5) rectangle (11.5, -1);
    \fill[myred!50, rounded corners] (11.5, -1.5) rectangle (12, -1);
    \draw[black!70, rounded corners] (10, -1.5) rectangle (10.5, -1);
    \draw[->] (9.75,0) -- (10.25,-1);
    \fill[black!50, rounded corners] (9.5, -1.5) rectangle (10, -1);
    
    \draw[->,thick] (4,-0.5) -- (5,-0.5);
    \end{tikzpicture}
    \caption{Example of an elementary move. Suppose that the error pattern indicates that $x_d = 1$, $x_e =1$, and the deletion does not reduce the number of runs while the substitution increases the number of runs by two. The process starts with the left figure in which a bit $1$ is inserted at position $\widetilde{d}$, the end of a $1$-run and the bit $1$ at position $\widetilde{e}-1$ is flipped. After an elementary move, $\widetilde{d}$ moves to the end of the next $1$-run, and $e$ moves to the next position that matches the error pattern $y_{\widetilde{e}-\widetilde{\delta}+1} = y_{\widetilde{e}-\widetilde{\delta}-1} = 0$.}
    \label{fig:elementary}
\end{figure}

We use different arguments according to how the number of runs changes to show that at most two valid candidate position pairs $(\widetilde{d},\widetilde{e})$ yield the correct value of $
f(x)$, $f_1^r(x)$, and $f_2^r(x)$ simultaneously. 
The following equations will be useful to determine how $\widetilde{d}$ and $\widetilde{e}$ change in each elementary move. 
Recall that we regard $y$ as a string obtained via one substitution at index $e-\delta$ from $x'\in\{0,1\}^{n-1}$, where $x'$ is obtained via one deletion from $x$ at index $d$. 
Note that
\begin{align*}
    &f(x) - f(x') = dx_d + \sum_{d}^{n-1}x'_i,
    &f(x') - f(y) = (e-\delta)[x_e - (1-x_e)].
\end{align*}
Moreover, we have
\begin{equation*}
    \sum_{d}^{n-1}x'_i=\sum_{d}^{n-1}y_i + \delta(2x_e-1).
\end{equation*}
Then, combining these three observations yields
\begin{equation}
    f(x)-f(y) = dx_d + \sum_{i=d}^{n-1} y_i + e(2x_e-1).\label{eqn:f(x)}
\end{equation}

\subsubsection{If the number of runs increases by two}
\label{sec:increase_two}

If $r(y) = r(x)+2$, then it must be that 
$r(x) = r(x')$ and $r(y) = r(x')+2$. 
This means that the deletion does not change the number of runs (and thus occurred in a run of length at least $2$ in $x$), while the substitution affects a bit in the middle of a run of length at least $3$.
In particular, we have $y_{e-\delta-1} = y_{e-\delta+1} = 1-y_{e-\delta}$. 
In this case, it follows that
\begin{align*}
    &f_1^r(x)-f_1^r(x') = r^x_d, & f_1^r(x')-f_1^r(y) = -(1+2(n-e+\delta)).
\end{align*}
Therefore, for the run-based sketch $f_1^r(x)$ it holds that
\begin{equation}
    f_1^r(x)-f_1^r(y) = r^x_d - (1 + 2(n-e+\delta)).\label{eqn:f1r_inc_2}
\end{equation}
We now proceed by case analysis on the value of $x_d$ and $x_e$.

\paragraph{If $x_e = x_d=b$:} In this case, when $\widetilde{d}$ makes an elementary move to the right, it must pass across a $(1-b)$-run of some length $\ell\geq 1$.
According to~\eqref{eqn:f(x)}, position $\widetilde{e}$ has to move to the left by $\ell$ so that
$f(\widetilde{x})=f(x)$.
If we have $\widetilde{d}<\widetilde{e}$ before one elementary move but $\widetilde{d}>\widetilde{e}$ after that move, we call it a \emph{take over step}. For each elementary move:

\begin{itemize}
    \item If the move is not a take over step: Then, $r^{\widetilde{x}}_{\widetilde{d}}$ increases by $2$ while $2(n-\widetilde{d}+\widetilde{\delta})+1$ increases by $2\ell$.
    Therefore,~\eqref{eqn:f1r_inc_2} implies that $f_1^r(\widetilde{x})$ strictly decreases after such a move whenever $\ell>1$.
    If $\ell=1$, then $\widetilde{e}$ moves by $1$ to the left and $f_1^r(\widetilde{x})$ remains unchanged. 
    However,
    since we need $1-b = y_{\widetilde{e}-\widetilde{\delta}} = 1-y_{\widetilde{e}-\widetilde{\delta}}$ it follows that $\widetilde{e}$ cannot move only $1$ position to the left, and so $\ell>1$ necessarily.
    See Figure~\ref{fig:elementary} for an example.
    
    \item If the move is a take over step: Before the move, $\widetilde{d}$ is on the left of a $(1-b)$-run of length $\ell\geq 1$ while $\widetilde{e}>\widetilde{d}$ satisfies $y_{\widetilde{e}-1}=1-b$ and $y_{\widetilde{e}-2} = y_{\widetilde{e}} = b$. 
    After the move, $\widetilde{d}$ moves to the right of the $(1-b)$-run of length $\ell$, while $\widetilde{e}$ is to the left of $\widetilde{d}$. 
    Moreover, it must be that $y_{\widetilde{e}}=1-b$ and $y_{\widetilde{e}-1} = y_{\widetilde{e}+1} = b$. 
    To match the error pattern, the only possible case is that $\ell = 1$. 
    To see why this is the case, note that when $\ell \geq 2$ the index $\widetilde{e}$ has to move to the left by at least $\ell+2$ to match the error pattern $y_{\widetilde{e}-\widetilde{\delta}-1} = y_{\widetilde{e}-\widetilde{\delta}+1} = 1-y_{\widetilde{e}-\widetilde{\delta}}$.
    However, this move leads to $f(\widetilde{x})\neq f(x)$, and thus does not yield a valid pair $(\widetilde{d}, \widetilde{e})$.
    When $\ell = 1$, let $(\widetilde{d}_1, \widetilde{e}_1)$ and $(\widetilde{d}_2, \widetilde{e}_2)$ denote the position pair before and after the move, respectively. 
    Then, these two pairs yield the same candidate solution $\widetilde{x}_1 = \widetilde{x}_2$. See Figure~\ref{fig:same_x} for an example.               
    
    \begin{figure}[h!]
        \centering
        \begin{tikzpicture}
    \draw (-0.25,-1.7) node [scale = 0.8] {$\widetilde{d}_1$};
    \draw (0.25,-1.75) node [scale = 0.8]{$\widetilde{e}_1$};
    \draw (-2.4,-1.25) node {$\widetilde{x}_1$};
    \fill[myblue!50, rounded corners] (-1.5, 0) rectangle (-0.5, 0.5);
    \fill[myred!50, rounded corners] (-0.5, 0) rectangle (0, 0.5);
    \fill[myblue!50, rounded corners] (0, 0) rectangle (1, 0.5);
    \draw[->] (-0.25,0) -- (0.25,-1);
    
    \draw (-2.4,0.24) node {$y$};
    \fill[myblue!50, rounded corners] (-1.5, -1.5) rectangle (1.5, -1);
    \draw[black!70, rounded corners] (0, -1.5) rectangle (0.5, -1);
    \fill[black!50, rounded corners] (-0.5, -1.5) rectangle (0, -1);
    
    \draw (3.6,0.24) node {$y$};
    \draw (5.75,-1.7) node [scale = 0.8] {$\widetilde{d}_2$};
    \draw (5.25,-1.75) node [scale = 0.8]{$\widetilde{e}_2$};
    \fill[myblue!50, rounded corners] (4, 0) rectangle (5, 0.5);
    \fill[myred!50, rounded corners] (5, 0) rectangle (5.5, 0.5);
    \fill[myblue!50, rounded corners] (5.5, 0) rectangle (6.5, 0.5);
    
    \draw (3.6,-1.25) node {$\widetilde{x}_2$};
    \fill[myblue!50, rounded corners] (4, -1.5) rectangle (7, -1);
    \draw[black!70, rounded corners] (5, -1.5) rectangle (5.5, -1);
    \draw[->] (5.25,0) -- (5.25,-1);
    \fill[black!50, rounded corners] (5.5, -1.5) rectangle (6, -1);
    
    \draw[->,thick] (2,-0.5) -- (3,-0.5);
    \end{tikzpicture}
        \caption{An example of a take over step. If the take over happens, it must be that $\ell = 1$. The resulting $\widetilde{x}_1$ and $\widetilde{x}_2$ are the same.}
        \label{fig:same_x}
    \end{figure}
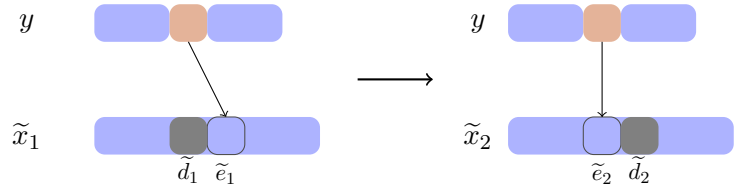

\end{itemize}

Taking into account both cases above, we see that $f_1^r(\widetilde{x})$ decreases during each elementary move, and decreases by at most $2n$ during the whole process.
Since the value of $f_1^r(x)$ is taken modulo $12n+1$, there is only a unique pair $(\widetilde{d}, \widetilde{e})$ that yields a solution such that $f_1^r(\widetilde{x}) = f_1^r(x)$.
Hence, $f(x)$ and $f_1^r(x)$ together with $y$ uniquely determine one valid pair $(\widetilde{d}, \widetilde{e})$, which in turn yields a unique candidate solution $\widetilde{x}=x$. 

\paragraph{If $x_d= 1-x_e = b$:} In this case, when $\widetilde{d}$ makes an elementary move to the right, it must pass across a
$(1-b)$-run of some length $\ell\geq 1$. 
Then, $\widetilde{e}$ has to move to the right by $\ell$ so that $f(\widetilde{x})=f(x)$. 
During each such move $f_1^r(\widetilde{x})$ strictly increases. For the whole process, $f_1^r(\widetilde{x})$ increases by at most $2n$. 
By a similar argument as above, we have that $f(x)$ and $f_1^r(x)$ together with $y$ uniquely determine one valid pair $(\widetilde{d}, \widetilde{e})$ which yields the correct solution $\widetilde{x} = x$.

\subsubsection{If the number of runs decreases by four}
\label{sec:decrease_four}
If $r(y) = r(x)-4$, then it must be that $r(x') = r(x)-2$ and $r(y) = r(x')-2$. 
This means that the error pattern must satisfy $x_{d-1} = x_{d+1} = 1-x_{d}$ and that $y_{e-\delta-1} = y_{e-\delta+1} = 1-y_{e-\delta}$.
In this case, since
\begin{align*}
    & f_1^r(x)-f_1^r(x') = r_d^x + 2(n+1-d),
    & f_1^r(x')-f_1^r(y) = 1 + 2(n-e+\delta),
\end{align*}
we have that
\begin{equation}
    f_1^r(x)- f_1^r(y) = [r_d^x + 2(n+1-d)] + [1+2(n-e+\delta)].\label{eqn:f1r_dec_4}
\end{equation}
We now proceed by case analysis on the value of $x_d$ and $x_e$.
\paragraph{If $x_d = x_e = b$:} We first find a solution of $\widetilde{d}$ and $\widetilde{e}$ such that $\widetilde{d}$ is as small as possible. 
During each elementary move, let $(\widetilde{d}_1, \widetilde{e}_1)$ denote the valid pair before the move and $(\widetilde{d}_2, \widetilde{e}_2)$ after the move. Let $\widetilde{x}_1$ and $\widetilde{x}_2$ be the resulting string, respectively. 
Recall that the rank of an index $i$ in a string $x$ is the $i$-th number in the run string $r^x$. For each elementary move,
\begin{itemize}
    \item If $\widetilde{d}_1 < \widetilde{e}_1$ and $\widetilde{d}_2 < \widetilde{e}_2$: 
    Since both the deletion and the substitution decreases the number of runs by two, for index $i$ such that $\widetilde{d}_1\leq i \leq \widetilde{d}_2$, the rank of index $i$ in $\widetilde{x}_1$ is larger than the rank of index $i$ in $\widetilde{x}_2$. 
    Similarly, for index $j$ such that $\widetilde{e}_2\leq j \leq \widetilde{e}_1$, the rank of index $j$ in $\widetilde{x}_1$ is smaller than the rank of index $j$ in $\widetilde{x}_2$. 
    By Lemma~\ref{lem:f2_convex}, we have that $f_2^r(\widetilde{x}_1) 
    \leq f_2^r(\widetilde{x}_2)$.
    Therefore, if $\widetilde{d} < \widetilde{e}$ before and after an elementary move, then $f_2^r(\widetilde{x})$ is strictly increasing.
    \item If $\widetilde{d}_1 > \widetilde{e}_1$ and $\widetilde{d}_2 > \widetilde{e}_2$: 
    For index $i$ such that $\widetilde{e}_2\leq i \leq \widetilde{e}_1$, the rank of index $i$ in $\widetilde{x}_1$ is smaller than the rank of index $i$ in $\widetilde{x}_2$. 
    Similarly, for index $j$ such that $\widetilde{d}_1\leq j \leq \widetilde{d}_2$, the rank of index $j$ in $\widetilde{x}_1$ is larger than the rank of index $j$ in $\widetilde{x}_2$. 
    By Lemma~\ref{lem:f2_convex}, we have that $f_2^r(\widetilde{x}_1) 
    \geq f_2^r(\widetilde{x}_2)$.
    Therefore, if $\widetilde{d} > \widetilde{e}$ before and after an elementary move, then $f_2^r(\widetilde{x})$ is strictly decreasing.
\end{itemize}
Since when $\widetilde{d}$ moves to the right, $\widetilde{e}$ has to move to the left accordingly to be a valid pair, we can have at most one take over step. 
Moreover, $f_2^r(\widetilde{x})$ increases by at most $4n^2$ if $\widetilde{d} < \widetilde{e}$ before and after the elementary move, and $f_2^r(\widetilde{x})$ decreases by at most $4n^2$ if $\widetilde{d} > \widetilde{e}$ before and after the elementary move. 
Since the value of $f_2^r(x)$ is taken modulo $16n^2+1$, 
this implies that we have at most two candidate position solutions $(\widetilde{d}_1,\widetilde{e}_1)$ and $(\widetilde{d}_2,\widetilde{e}_2)$, where $\widetilde{d}_1 < \widetilde{e}_1$ and $\widetilde{d}_2 > \widetilde{e}_2$ such that
$f(\widetilde{x}_1)=f(\widetilde{x}_2)=f(x)$ and $f_2^r(\widetilde{x}_1)=f_2^r(\widetilde{x}_2)=f_2^r(x)$. 
Hence, $f(x)$, $f_1^r(x)$, and $f_2^r(x)$ together with $y$ yield at most two candidate position solutions $(\widetilde{d}_1,\widetilde{e}_1)$ and $(\widetilde{d}_2,\widetilde{e}_2)$, and thus at most two candidate solutions $\widetilde{x}_1$ and $\widetilde{x}_2$.

\paragraph{If $x_d = 1-x_e = b$:} We first find a valid pair $(\widetilde{d},\widetilde{e})$ such that $\widetilde{d}$ is as small as possible. 
When $\widetilde{d}$ makes an elementary move to the right, $\widetilde{e}$ needs to move to the right such that $f(\widetilde{x}) = f(x)$. 
During such moves, if $\widetilde{d}$ moves to the right by $\ell$, then $r^{\widetilde{x}}_{\widetilde{d}}$ increases by at most $\ell$, while $2(n+1-d)$ decreases by $2\ell$. 
Moreover, $2(n-e+\delta)$ is non-increasing. 
Therefore, by~\eqref{eqn:f1r_dec_4}, $f_1^r(\widetilde{x})$ strictly decreases during the elementary moves.
For the whole process, $f_1^r(\widetilde{x})$ decreases by at most $4n$.
Hence, following a similar argument as above, $f(x)$ and $f_1^r(x)$ together with $y$ uniquely determine a position pair $(\widetilde{d}, \widetilde{e})$, and thus a unique candidate solution $\widetilde{x}$.

\subsubsection{If the number of runs decreases by two}
\label{sec:decrease_two}
If $r(y) = r(x)-2$, it might be that the substitution decreases the number of runs by two: $r(x') = r(x)$ and $r(y) = r(x')-2$; or that the deletion decreases the number of runs by two: $r(x')=r(x)-2$ and $r(y) = r(x')$. 
We will treat these two sub-cases separately, and we will show that in each of these sub-cases we can uniquely recover $x$.

\paragraph{Substitution decreases the number of runs by two.}
We consider the case in which the substitution decreases the number of runs by two, i.e.,  $r(x') = r(x)$ and $r(y) = r(x')-2$. 
Since the substitution decreases the number of runs, it must flip a bit $b$ in a $b$-run of length one, i.e., $y_{e-\delta-1} = y_{e-\delta+1} = y_{e-\delta}$, and we also have
\begin{equation}
    f_1^r(x)-f_1^r(y) = r^x_d + (1+2(n-e+\delta)).\label{eqn:f1r_dec_2_sub}
\end{equation}
We will proceed by case analysis based on the value of $x_d$ and $x_e$.

\paragraph{If $x_d = x_e = b$:} 
We first find a valid pair $(\widetilde{d},\widetilde{e})$ such that $\widetilde{d}$ is as small as possible. 
When $\widetilde{d}$ makes an elementary move to the right, it must pass across a $(1-b)$-run of length $\ell\geq 1$. 
Henceforth, in this elementary move $\widetilde{e}$ must to move to the left by $\ell$ so that $f(\widetilde{x})=f(x)$, according to~\eqref{eqn:f(x)}.
During such moves, $r_{\widetilde{d}}^{\widetilde{x}}$ strictly increases, while $2(n-e+\delta)$ is non-decreasing.
As a result, by~\eqref{eqn:f1r_dec_2_sub}, $f_1^r(\widetilde{x})$ increases during each elementary move.
For the whole process, $f_1^r(x)$ increases by at most $2n$.
Therefore, analogously to previous cases, $f(x)$ and $f_1^r(x)$ together with $y$ uniquely determine a valid position pair $(\widetilde{d}, \widetilde{e})$, and thus a unique candidate solution $\widetilde{x}=x$.

\paragraph{If $x_d = 1 - x_e = b$:} 
We first find a valid pair $(\widetilde{d},\widetilde{e})$ such that $\widetilde{d}$ is as small as possible. When $\widetilde{d}$ makes an elementary move to the right, it must pass across a $(1-b)$-run of length $\ell\geq 1$. According to ~\eqref{eqn:f(x)}, $\widetilde{e}$ needs to move to the right by $\ell$ to ensure that $f(\widetilde{x}) = f(x)$. 
\begin{itemize}
    \item If this is not a take over step: Suppose that the valid pairs are $(\widetilde{d}_1,\widetilde{e}_1)$ and $(\widetilde{d}_2, \widetilde{e}_2)$ before and after the move, respectively. 
    During each such move, the run number $r^{\widetilde{x}}_{\widetilde{d}}$ increases by $2$ while $[1+2(n-\widetilde{e}+\widetilde{\delta})]$ decreases by $2\ell$. By~\eqref{eqn:f1r_dec_2_sub}, if $\ell > 1$, then $f_1^r(\widetilde{x})$ decreases; 
    if $\ell = 1$, then $f_1^r(\widetilde{x})$ does not change.
    
    However, when $\ell = 1$, the rank of $\widetilde{d}_1$ and $\widetilde{d}_1+1$ in $\widetilde{x}_1$ is smaller than that in $\widetilde{x}_2$, while the rank of $\widetilde{e}_1$ and $\widetilde{e}_2$ in $\widetilde{x}_1$ is larger than that in $\widetilde{x}_2$.
    By Lemma~\ref{lem:f2_convex}, $f_2^r(\widetilde{x}_1) > f_2^r(\widetilde{x}_2)$. 
    This implies that during such elementary moves, either $f_1^r(\widetilde{x})$ strictly decreases, or $f_2^r(\widetilde{x})$ strictly decreases.
    \item If this is a take over step: During each such move, the run number $r^{\widetilde{x}}_{\widetilde{d}}$ increases by $2$ while $[1+2(n-\widetilde{e}+\widetilde{\delta})]$ decreases by $2\ell+2$. 
    Therefore, $f_1^r(\widetilde{x})$ decreases.
\end{itemize}
During the whole process, $f_1^r(x)$ decreases by at most $2n$, while $f_2^r(x)$ decreases by at most $4n^2$. 
Consequently, $f(x)$ and $f_1^r(x)$ together with $y$ uniquely determine a valid pair $(\widetilde{d},\widetilde{e})$, and thus a unique candidate solution $\widetilde{x}=x$.

\paragraph{Deletion decreases the number of runs by two.}
In this section we consider the case in which $r(x')=r(x)-2$ and $r(y) = r(x')$. This means that the substitution happens at the beginning or end of a run. Moreover, the deletion pattern satisfies  $x_{d-1} = x_{d+1} = 1-x_d$. Let $\gamma$ be an indicator variable such that $\gamma = 1$ if the substitution is at the end of a run in $y$, and $\gamma = -1$ if the substitution is at the beginning of a run in $y$. Then we have that 
\begin{equation}
    f_1^r(x)-f_1^r(y) = r^x_d+2(n+1-d)+\gamma.\label{eqn:f1r_dec_2_del}
\end{equation}
We now proceed by case analysis cased on the value of $x_d$ and $x_e$.

\paragraph{If $x_d = x_e = b$:} 
We first find a valid pair $(\widetilde{d},\widetilde{e})$ such that $\widetilde{d}$ is as small as possible. 
There are two kinds of elementary moves of $\widetilde{d}$. 
\begin{itemize}
    \item The index $\widetilde{d}$ moves within a $(1-b)$-run of length at least two: In this case, when $\widetilde{d}$ moves to the right by $\ell$, $\widetilde{e}$ has to move to the left by $\ell$ so that $f(\widetilde{x})=f(x)$. In such a move, $r_{\widetilde{d}}^{\widetilde{d}}$ does not change, while $2(n+1-d)$ decreases by $2\ell$. According to~\eqref{eqn:f1r_dec_2_del}, $f_1^r(\widetilde{x})$ decreases if $\ell > 1$.
    When $\ell = 1$, if $\gamma$ changes from $-1$ to $1$, the run-based sketch $f_1^r(\widetilde{x})$ remains unchanged by~\eqref{eqn:f1r_dec_2_del}.
    This implies that during this move, $\widetilde{e}$ moves from the beginning of a $b$-run in $y$ to the end of a $b$-run in $y$.
    Otherwise, the substitution will affect the number of runs.
    However, $\widetilde{e}$ has to move to the left by at least two to match the error pattern. Thus, it must be that $\ell > 1$.
    
    \item The index $\widetilde{d}$ moves across a $b$-run of length $\ell$: 
    In this case, $\widetilde{d}$ moves to the right by $\ell+\ell'$ for $\ell'\geq 2$, $\widetilde{e}$ has to move to the left by $\ell'$ such that $f(\widetilde{x}) = f(x)$. 
    During such moves, $f_1^r(\widetilde{x})$ decreases.
\end{itemize}
For the whole process, $f_1^r(x)$ decreases by at most $2n$.
Therefore, $f(x)$ and $f_1^r(x)$, together with $y$, uniquely determine one position pair $(\widetilde{d}, \widetilde{e})$, and thus a unique candidate solution $\widetilde{x}$.

\paragraph{If $x_d = 1-x_e = b$:} We first find a valid pair $(\widetilde{d}, \widetilde{e})$ such that $\widetilde{d}$ is as small as possible. There are two kinds of elementary moves of $\widetilde{d}$. 

\begin{itemize}
    \item The first one is that when $\widetilde{d}$ moves within a $(1-b)$-run of length at least two. In this case, when $\widetilde{d}$ moves to the right by $\ell$, $\widetilde{e}$ has to move to the right by $\ell$ such that $f(\widetilde{x}) = f(x)$. During such moves, $f_1^r(\widetilde{x})$ is non-increasing. Moreover, $f_1^r(\widetilde{x})$ only remains unchanged during such moves if $\ell=1$ and $\widetilde{e}$ moves from the beginning of a $b$-run of length two to the end of this $b$-run; otherwise the substitution will affect the number of runs. During such a move, $f_2^r(\widetilde{x})$ is increasing.
    \item The second one is that when $\widetilde{d}$ moves cross a $b$-run of length $\ell$. 
    In this case, $\widetilde{d}$ moves to the right by $\ell+\ell'$ for some $\ell'\geq 2$, $\widetilde{e}$ has to move to the right by $\ell'$ to match $f(x)$. 
    During such moves, $f_1^r(\widetilde{x})$ decreases.
\end{itemize}

For the whole process, $f_1^r(x)$ decreases by at most $2n$, while $f_2^r(x)$ increases by at most $4n^2$.
Therefore, $f(x)$, $f_1^r(x)$ and $f_2^r(x)$, together with $y$, uniquely determine a position pair $(\widetilde{d}, \widetilde{e})$, and thus a unique candidate solution $\widetilde{x}$.

\subsubsection{If the number of runs does not change}
\label{sec:unchange}
If $r(y) = r(x)$, it might be that both errors do not change the number of runs: $r(x') = r(x)$ and $r(y) = r(x')$; or that deletion reduces two runs while substitution increases two runs: $r(x')=r(x)-2$ and $r(y) = r(x')+2$. 

\paragraph{Both errors do not change the number of runs.}
In this case, let $\gamma$ be an indicator variable satisfying $\gamma = 1$ if the
substitution is at the end of a run in $y$, and $\gamma = -1$ if the substitution is at the beginning of a run in $y$. We have
\begin{equation}
    f_1^r(x) - f_1^r(y) = r_d^x+\gamma. \label{eqn:f1r_unchange}
\end{equation}
Moreover,
\begin{equation}
    f_2^r(x)-f_2^r(y) = 
    \begin{cases}
    r_d^x(r_d^x-1) + 2 (r^x_{e}-1), \quad\text{ if } \gamma = 1,\\
    r_d^x(r_d^x-1) - 2 r^x_{e}, \quad\text{ if } \gamma = -1.
    \end{cases}
    \label{eqn:f2r_unchange}
\end{equation}
We now proceed by case analysis cased on the value of $x_d$ and $x_e$.

\paragraph{If $x_d = x_e = b$:} We have at most two solutions, one with $\gamma = 1$ and one with $\gamma = -1$. Assume that the pair $(\widetilde{d}_1,\widetilde{e}_1)$ corresponding to $\gamma = 1$ yields codeword $\widetilde{x}_1$, and that the pair $(\widetilde{d}_2,\widetilde{e}_2)$ corresponding to $\gamma = -1$ yields codeword $\widetilde{x}_2$. We now show that if $f(\widetilde{x}_1) = f(\widetilde{x}_2)$, $f_1^r(\widetilde{x}_1) = f_1^r(\widetilde{x}_2)$, and that $f_2^r(\widetilde{x}_1) = f_2^r(\widetilde{x}_2)$, then $\widetilde{x}_1$ must be same as $\widetilde{x}_2$.

Note that since  $f_1^r(\widetilde{x}_1) = f_1^r(\widetilde{x}_2)$and   $f_2^r(\widetilde{x}_1) = f_2^r(\widetilde{x}_2)$, by~\eqref{eqn:f1r_unchange} and~\eqref{eqn:f2r_unchange}, we have
\begin{align*}
    &r_{\widetilde{d}_2}^{\widetilde{x}_2} = r_{\widetilde{d}_1}^{\widetilde{x}_1} + 2,
    &r_{\widetilde{d}_1}^{\widetilde{x}_1}(r_{\widetilde{d}_1}^{\widetilde{x}_1}-1) + 2 (r^{\widetilde{x}_1}_{\widetilde{e}_1}-1) = r_{\widetilde{d}_2}^{\widetilde{x}_2}(r_{\widetilde{d}_2}^{\widetilde{x}_2}-1) - 2 r_{\widetilde{e}_2}^{\widetilde{x}_2}.
\end{align*}
This means that $r_{\widetilde{e}_1}^{\widetilde{x}_1} + r_{\widetilde{e}_2}^{\widetilde{x}_2} = r_{\widetilde{d}_1}^{\widetilde{x}_1} + r_{\widetilde{d}_2}^{\widetilde{x}_2}$, or, equivalently,
\begin{equation}
    r_{\widetilde{e}_1-\widetilde{\delta}_1}^y + r_{\widetilde{e}_2-\widetilde{\delta}_2}^{y} = r_{\widetilde{d}_1-1}^y + r_{\widetilde{d}_2-1}^{y}, \label{eqn:unchange_equal}
\end{equation}
where $\delta_1$ ($\delta_2$) is the indicator of whether $\widetilde{d}_1$ is smaller than $\widetilde{e}_1$ ($\widetilde{d}_2$ is smaller than $\widetilde{e}_2$, respectively).
Since both $\widetilde{d}_1$ and $\widetilde{d}_2$ do not change the runs, there exists a single $(1-b)$-run of length $\ell$ between $\widetilde{d}_1$ and $\widetilde{d}_2$ in $y$, and that $\widetilde{e}_2 - \widetilde{e}_1 = \ell$. 
\begin{itemize}
    \item If $r^y_{\widetilde{e}_1-\widetilde{\delta}_1} > r_{\widetilde{d}_2-1}^{y}$: By~\eqref{eqn:unchange_equal}, it must be that $r_{\widetilde{e}_2-\widetilde{\delta}_2}^{y} < r_{\widetilde{d}_1-1}^y$. That is to say,
    \begin{equation*}
        \widetilde{e}_2-\widetilde{\delta}_2< \widetilde{d}_1-1 < \widetilde{d}_2-1 < \widetilde{e}_1-\widetilde{\delta}_1.
    \end{equation*} 
    However, this contradicts the fact that $\widetilde{e}_2 - \widetilde{e}_1 = \ell$.
    \item If $r^y_{\widetilde{e}_1-\widetilde{\delta}_1} = r_{\widetilde{d}_2-1}^{y}$: This is impossible since $y_{\widetilde{e}_1-\widetilde{\delta}_1} = 1-b$ while $y_{\widetilde{d}_2-1} = b$.

    \item If $r^y_{\widetilde{e}_1-\widetilde{\delta}_1} < r_{\widetilde{d}_2-1}^{y}$: By~\eqref{eqn:unchange_equal} and the fact that $f(\widetilde{x}_1) = f(\widetilde{x}_2)$, it must be that $r_{\widetilde{e}_1-\widetilde{\delta}_1}^y = r_{\widetilde{d}_1-1}^y+1$, and $r_{\widetilde{e}_2-\widetilde{\delta}_2}^y = r_{\widetilde{d}_2-1}^y-1$, i.e., $\widetilde{e}_1-\widetilde{\delta}_1$ is the end of the $(r^y_{\widetilde{d}_1-1}+1)$-th run in $y$, and $\widetilde{e}_2-\widetilde{\delta}_2$ is the beginning of this run. In this case, we must have $\widetilde{x}_1=\widetilde{x}_2$. 
    See Figure~\ref{fig:same_gamma} for an example.
    
    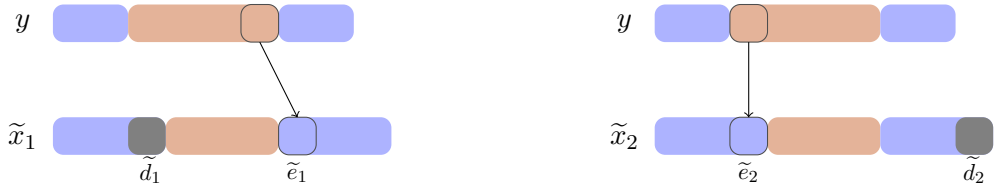
\begin{figure}[h!]
    \centering
    \begin{tikzpicture}
    \draw (-2.4,0.24) node {$y$};
    \draw (-0.7,-1.7) node [scale = 0.8] {$\widetilde{d}_1$};
    \draw (1.25,-1.73) node [scale = 0.8]{$\widetilde{e}_1$};
    \fill[myblue!50, rounded corners] (-2, -1.5) rectangle (-0.5, -1);
    \fill[myred!50, rounded corners] (-0.5, -1.5) rectangle (1, -1);
    \fill[myblue!50, rounded corners] (1, -1.5) rectangle (2.5, -1);
    
    \fill[black!50, rounded corners] (-1, -1.5) rectangle (-0.5, -1);
    \draw[black!70, rounded corners] (1, -1.5) rectangle (1.5, -1);
    
    \draw (-2.4,-1.25) node {$\widetilde{x}_1$};
    \fill[myblue!50, rounded corners] (-2, 0) rectangle (-1, 0.5);
    \fill[myred!50, rounded corners] (-1, 0) rectangle (1, 0.5);
    \fill[myblue!50, rounded corners] (1, 0) rectangle (2, 0.5);
    \draw[black!70, rounded corners] (0.5, 0) rectangle (1, 0.5);
    \draw[->] (0.75,0) -- (1.25,-1);
    
    \draw (5.6,0.24) node {$y$};
    \draw (10.25,-1.7) node [scale = 0.8] {$\widetilde{d}_2$};
    \draw (7.25 ,-1.73) node [scale = 0.8]{$\widetilde{e}_2$};
    \fill[myblue!50, rounded corners] (6, -1.5) rectangle (7.5, -1);
    \fill[myred!50, rounded corners] (7.5, -1.5) rectangle (9, -1);
    \fill[myblue!50, rounded corners] (9, -1.5) rectangle (10.5, -1);
    
    \fill[black!50, rounded corners] (10, -1.5) rectangle (10.5, -1);
    \draw[black!70, rounded corners] (7, -1.5) rectangle (7.5, -1);
    \draw (5.6,-1.25) node {$\widetilde{x}_2$};
    \fill[myblue!50, rounded corners] (6, 0) rectangle (7, 0.5);
    \fill[myred!50, rounded corners] (7, 0) rectangle (9, 0.5);
    \fill[myblue!50, rounded corners] (9, 0) rectangle (10, 0.5);
    \draw[black!70, rounded corners] (7, 0) rectangle (7.5, 0.5);
    \draw[->] (7.25,0) -- (7.25,-1);
    \end{tikzpicture}
    \caption{An example of $\widetilde{x}_1 =\widetilde{x}_2$, where $\widetilde{x}_1$ is a candidate solution with $\gamma = 1$ while $\widetilde{x}_2$ is a candidate solution with $\gamma = -1$.}
    \label{fig:same_gamma}
\end{figure}

\end{itemize}
Therefore, in this case, $f(x)$, $f_1^r(x)$, and $f_2^r(x)$ together with $y$ uniquely determine a valid pair $(\widetilde{d}, \widetilde{e})$ and thus a unique candidate solution $\widetilde{x}=x$.

\paragraph{If $x_d = 1-x_e = b$:} 
We have at most two solutions, one with $\gamma = 1$ and one with $\gamma = -1$. Assume that the pair $(\widetilde{d}_1,\widetilde{e}_1)$ corresponding to $\gamma = 1$ yields codeword $\widetilde{x}_1$, and that the pair $(\widetilde{d}_2,\widetilde{e}_2)$ corresponding to $\gamma = -1$ yields codeword $\widetilde{x}_2$.
By a similar argument as above, it must be that $r_{\widetilde{e}_1-\widetilde{\delta}_1}^y = r^y_{\widetilde{d}_1-1}$, $r_{\widetilde{e}_2-\widetilde{\delta}_2}^y = r^y_{\widetilde{d}_2-1}$. These two position pairs actually yield the same result $\widetilde{x}_1 = \widetilde{x}_2$. See Figure~\ref{fig:same_gamma_2} for an example.

    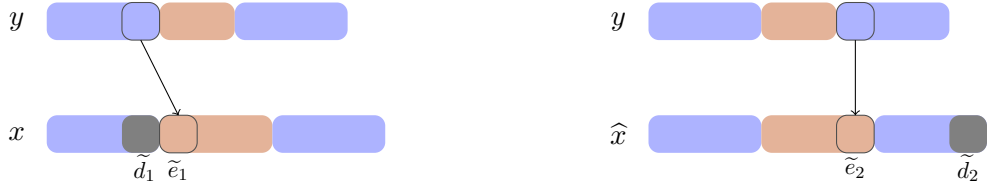
\begin{figure}[h!]
    \centering
    \begin{tikzpicture}
    \draw (-2.4,0.24) node {$y$};
    \draw (-0.7,-1.7) node [scale = 0.8] {$\widetilde{d}_1$};
    \draw (-0.25,-1.73) node [scale = 0.8]{$\widetilde{e}_1$};
    \fill[myblue!50, rounded corners] (-2, -1.5) rectangle (-0.5, -1);
    \fill[myred!50, rounded corners] (-0.5, -1.5) rectangle (1, -1);
    \fill[myblue!50, rounded corners] (1, -1.5) rectangle (2.5, -1);
    
    \fill[black!50, rounded corners] (-1, -1.5) rectangle (-0.5, -1);
    \draw[black!70, rounded corners] (-0.5, -1.5) rectangle (0, -1);
    
    \draw (-2.4,-1.25) node {$x$};
    \fill[myblue!50, rounded corners] (-2, 0) rectangle (-0.5, 0.5);
    \fill[myred!50, rounded corners] (-0.5, 0) rectangle (0.5, 0.5);
    \fill[myblue!50, rounded corners] (0.5, 0) rectangle (2, 0.5);
    \draw[black!70, rounded corners] (-1, 0) rectangle (-0.5, 0.5);
    \draw[->] (-0.75,0) -- (-0.25,-1);
    
    \draw (5.6,0.24) node {$y$};
    \draw (10.25,-1.7) node [scale = 0.8] {$\widetilde{d}_2$};
    \draw (8.75 ,-1.7) node [scale = 0.8]{$\widetilde{e}_2$};
    \fill[myblue!50, rounded corners] (6, -1.5) rectangle (7.5, -1);
    \fill[myred!50, rounded corners] (7.5, -1.5) rectangle (9, -1);
    \fill[myblue!50, rounded corners] (9, -1.5) rectangle (10.5, -1);
    
    \fill[black!50, rounded corners] (10, -1.5) rectangle (10.5, -1);
    \draw[black!70, rounded corners] (8.5, -1.5) rectangle (9, -1);
    \draw (5.6,-1.25) node {$\widehat{x}$};
    \fill[myblue!50, rounded corners] (6, 0) rectangle (7.5, 0.5);
    \fill[myred!50, rounded corners] (7.5, 0) rectangle (8.5, 0.5);
    \fill[myblue!50, rounded corners] (8.5, 0) rectangle (10, 0.5);
    \draw[black!70, rounded corners] (8.5, 0) rectangle (9, 0.5);
    \draw[->] (8.75,0) -- (8.75,-1);
    \end{tikzpicture}
    \caption{An example of $\widetilde{x}_1 = \widetilde{x}_2$, where $\widetilde{x}_1$ is a candidate solution with $\gamma = 1$ while $\widetilde{x}_2$ is a candidate solution with $\gamma = -1$.}
    \label{fig:same_gamma_2}
\end{figure}
Therefore, $f(x)$, $f_1^r(x)$, and $f_2^r(x)$ together with $y$ uniquely determine a valid pair $(\widetilde{d}, \widetilde{e})$, and thus a unique candidate solution $\widetilde{x}=x$.

\paragraph{Deletion reduces the number of runs while the substitution increases the number of runs.}

In this case, we have
\begin{align*}
    & f_1^r(x) - f_1^r(x') = r_d^x + 2(n+1-d),
    & f_1^r(x') - f_1^r(y) = -1 - 2(n-e).
\end{align*}
Therefore,
\begin{equation}
    f_1^r(x) - f_1^r(y) = [r_d^x + 2(n+1-d)]-[1+2(n-e+\delta)].\label{eqn:f1r_unchanged_2}
\end{equation}
We now proceed by case analysis cased on the value of $x_d$ and $x_e$.
\paragraph{If $x_d = x_e = b$:} We first find a valid pair of $\widetilde{d}$ and $\widetilde{e}$ such that $\widetilde{d}$ is as small as possible. When $\widetilde{d}$ makes an elementary move to the right, $\widetilde{e}$ has to move to the left to match $f(x)$. During such moves, 
$r_{\widetilde{d}}^{\widetilde{x}}+2(n+1-d)$ decreases, while $1+2(n-e+\delta)$ is non-decreasing.
By~\eqref{eqn:f1r_unchanged_2}, $f_1^r(\widetilde{x})$ decreases. 
During this process, $f_1^r(x)$ decreases by at most $4n$. By a similar argument as above, $f(x)$, $f_1^r(x)$, together with $y$ uniquely determine a position pair $(\widetilde{d}, \widetilde{e})$, and thus a unique candidate solution $\widetilde{x}$.
    
\paragraph{If $x_d = 1-x_e = d$:} We first find a valid pair $(\widetilde{d}, \widetilde{e})$ such that $\widetilde{d}$ is as small as possible. 
When $\widetilde{d}$ moves to the right, $\widetilde{e}$ has to move to the right to make sure than $f(\widetilde{x}) = f(x)$. 
During such moves, $f_2^r(\widetilde{d})$ increases if $\widetilde{e} > \widetilde{d}$, and decreases if $\widetilde{e} < \widetilde{d}$ by a similar argument in Section~\ref{sec:decrease_four}. 
Therefore, we have at most two solutions, $(\widetilde{d}_1,\widetilde{e}_1)$ and $(\widetilde{d}_2,\widetilde{e}_2)$, where $\widetilde{d}_1 < \widetilde{e}_1$, $\widetilde{d}_2 > \widetilde{e}_2$, such that $f(\widetilde{x}_1) = f(\widetilde{x}_2)=f(x)$, and meanwhile, $f_2^r(\widetilde{x}_1)= f_2^r(\widetilde{x}_2)=f_2^r(x)$.
Since $f_2^r(x)$ increases by at most $4n^2$ when $\widetilde{e} > \widetilde{d}$ and decreases by at most $4n^2$ when $\widetilde{e} < \widetilde{d}$, $f(x)$ and $f_2^r(x)$ together determine at most two possible position pairs $(\widetilde{d}_1,\widetilde{e}_1)$ and $(\widetilde{d}_2,\widetilde{e}_2)$, and thus at most two candidate solutions $\widetilde{x}_1$ and $\widetilde{x}_2$.

\subsection{A linear-time decoder}
\label{sec:delandsub_decoder}

We now proceed to describe our decoding procedure for $\cC$ based on the case analysis described in Section~\ref{sec:correct1del1sub}:

\begin{enumerate}
    \item If $|y| = n$, i.e., there is no deletion, we use the VT sketch $f(x)$ to decode directly.
    \item If $|y| = n-1$, i.e., the error includes one deletion at some index $x$ and one substitution at index $e$, we check $h_r(x)-h_r(y)$.
    By comparing $h(x)$ and $h(y)$ we can also recover the value of $x_d$ and $x_e$ according to Table~\ref{tab:h(x)}.
    \begin{enumerate}
     \item If the number of runs increases by two ($h_r(x)-h_r(y) = -2$), the analysis in Section~\ref{sec:increase_two} shows that we can recover a unique candidate solution $\widetilde{x}=x$ that matches all the sketches simultaneously in linear time.
    
    \item If the number of runs decreases by four ($h_r(x)-h_r(y) = 4$), the analysis in Section~\ref{sec:decrease_four} implies we can recover at most two candidate solutions $\widetilde{x}_1$ and $\widetilde{x}_2$ that match all the sketches simultaneously in linear time, and we are guaranteed that $x\in\{\widetilde{x}_1,\widetilde{x}_2\}$. 
    
    \item
    If the number of runs increases by two ($h_r(x)-h_r(y) = 2$), the analysis in  Section~\ref{sec:decrease_two} implies that we can uniquely recover $x$ in linear time \emph{if} we know which of the errors (deletion or substitution) reduced the number of runs by $2$. This property yields a linear time list-size $2$ decoder as follows. We run the decoder for the two different cases above on $y$.
    We are guaranteed that one of the decoders will behave correctly and output $x$.
    The other decoder may behave arbitrarily, but we know that if it outputs more than one (possibly erroneous) candidate string then it is not the correct decoder, and we can then disregard its output.
    Therefore, in the worst case we obtain a list of size $2$ containing $x$.

    \item 
    If the number of runs does not change ($h_r(x)=h_r(y)$), the analysis in  Section~\ref{sec:decrease_four} implies that we can uniquely recover $x$ in linear time \emph{if} we know whether both errors did not affect the number of runs or whether the deletion reduced the number of runs by $2$ while the substitution increased it by $2$.
    By an analogous argument to the previous item where we run both decoders, this property implies that we can recover a list of size $2$ containing $x$ in linear time.

    \end{enumerate}
\end{enumerate}

Therefore, we can list decode from one deletion and one substitution with a list of size at most two in time $O(n)$.

\subsection{A linear-time encoder}
In Section~\ref{sec:correct1del1sub} we gave a list decoding procedure that corrects one deletion and at most one substitution given knowledge of the VT sketch~\eqref{eq:vt1del1sub}, the run-based sketches~\eqref{eq:run1} and~\eqref{eq:run2}, and the count sketches~\eqref{eq:count1delsub} and~\eqref{eq:count2delsub}.
In this section, we describe a linear-time encoding procedure for a slightly modified version of the code $\cC$ defined in~\eqref{eq:codedelsub} with redundancy $4\log n + O(\log\log n)$ which inherits the same list decoding procedure and properties from Section~\ref{sec:delandsub_decoder}.
This approach is standard and very similar to Section~\ref{sec:linenc}.

Consider an arbitrary input string $x\in\bits^m$ for some fixed message length $m$.
Let $(\overline{\Enc}, \overline{\Dec})$ denote the efficient encoding and decoding procedures of a code for messages of length 
\[\ell = |f(x)\|f_1^r(x)\|f_2^r(x)\|h(x)\|h_r(x)|\]
correcting one deletion and one substitution (here, we represent the sketches via their binary representations).
For example, we may take the code from~\cite{SWWY20}, which has redundancy $6\log \ell + 8 = O(\log \log m)$.
Let 
\[u= \overline{\Enc}(f(x)\|f_1^r(x)\|f_2^r(x)\|h(x)\|h_r(x)).\] 
Then, we take final encoding procedure $\Enc$ to be
\begin{equation*}
    \Enc(x)= x\|u\in\bits^n,
\end{equation*}
which runs in time $O(m)=O(n)$ with overall redundancy $|u|=4\log n+O(\log\log n)$.

We now describe a linear-time decoding procedure.
Suppose that $\Enc(x)$ is corrupted into a string $y$ via at most one deletion and one substitution.
First, note that we can recover $u$ by running $\overline{\Dec}$ on the last $|u|-1$ bits of $y$.
Then, using the linear-time decoding procedure described in Section~\ref{sec:delandsub_decoder}, we can recover a list of size at most two containing $x$ from $u$ (which encodes the necessary sketches) and $y'=y[1:m-1]$.
This yields Theorem~\ref{thm:delandsub}.

\section{Binary codes correcting one deletion or one transposition}\label{sec:deltrans}

We prove Theorem~\ref{thm:deltrans} in this section.
Our starting point is a marker-based segmentation approach considered by Lenz and Polyanskii~\cite{LP20} to correct bursts of deletions.
We then introduce several new ideas.
Roughly speaking, our idea is to partition a string $x\in\bits^n$ into consecutive short substrings $z^x_1,\dots,z^x_\ell$ for some $\ell$ according to the occurrences of a special marker string in $x$.
Then, by carefully embedding hashes of each segment $z^x_i$ into a VT-type sketch and exploiting specific structural properties of deletions and adjacent transpositions, we are able to determine a short interval containing the position where the error occurred.
Once this is done, a standard technique allows us to recover the true position of the error by slightly increasing the redundancy.

\subsection{Code construction}

We now describe the code construction in detail.
For a given integer $n>0$, let $\Delta=50+1000\log n$ and $m=1000\Delta^2=O(\log^2 n)$.
For the sake of readability,
we have made no efforts to optimize constants, and assume $n$ is a power of two to avoid using ceilings and floors.
Given a string $x\in\bits^n$, we divide it into substrings split according to occurrences of the marker $0011$.
To avoid edge cases, assume that $x$ ends in $0011$ -- this will only add $4$ bits to the overall redundancy.
Then, this marker-based segmentation induces a vector
\begin{equation*}
    z^x = (z^x_1,\dots,z^x_{\ell_x}),
\end{equation*}
where $1\leq \ell_x\leq n$, and each string $z^x_i$ has length at least $4$, ends with $0011$, and $0011$ only occurs once in each such string.
We may assume that $|z^x_i|\leq \Delta$ for all $i$. This will only add $1$ bit to the overall redundancy, as captured in the following simple lemma.
\begin{lem}\label{lem:shortsegment}
    Suppose $X$ is uniformly random over $\bits^n$.
    Then, we have
    \begin{equation*}
        \Pr[|z^X_i|\leq \Delta, i=1,\dots,\ell_X]\geq \frac{1}{2}.
    \end{equation*}
\end{lem}
\begin{proof}
Since the probability that a fixed length-$4$ substring of $X$ equals $0011$ is $1/16$, it follows that the probability that $|z^X_i|>\Delta$ for any fixed $i$ is at most
\begin{equation*}
    \left(\frac{15}{16}\right)^{\Delta/4-1}\leq \frac{1}{2n^3}.
\end{equation*}
A union bound over all $1\leq i\leq n$ yields the desired statement.
\end{proof}

Our goal now will be to impose constraints on $z^x$ so that (i) We only introduce $\log n+O(\log\log n)$ bits of redundancy, and (ii) If $x$ is corrupted by a deletion or transposition in $z^x_i$, we can then locate a window $W\subseteq[n]$ of size $|W|=O(\log^4 n)$ such that $z^x_i\subseteq W$.
This will then allows us to correct the error later on by adding $O(\log\log n)$ bits of redundancy.

Since each $z^x_i$ has length at most $\Delta=O(\log n)$, we will exploit the fact that there exists a hash function $h$ with short output that allows us to correct a deletion, substitution, or transposition in all strings of length at most $3\Delta$.
This is guaranteed by the following lemma.
\begin{lem}\label{lem:hash}
    There exists a hash function $h:\bits^{\leq 3\Delta}\to [m]$ with the following property: \emph{If $z'$ is obtained from $z$ by at most two transpositions, two substitutions, or at most a deletion and an insertion, then $h(z)\neq h(z')$.}
\end{lem}
\begin{proof}
We can construct such a hash function $h$ greedily.
Let $A(z)$ denote the set of such strings obtained from $z\in\bits^{\leq 3\Delta}$.
Since $|A(z)|<m$, we can set $h(z)$ so that $h(z)\neq h(z')$ for all $z'\in A(z)\setminus\{z\}$.
\end{proof}

With the intuition above and the hash function $h$ guaranteed by Lemma~\ref{lem:hash} in mind, we consider the VT-type sketch
\begin{equation*}
    f(x) = \sum_{j=1}^{\ell_x} j(|z^x_j|\cdot m+h(z^x_j)) \mod (L=10n\cdot \Delta\cdot m +1)
\end{equation*}
along with the count sketches
\begin{align*}
    &g_1(x)=\ell_x\mod 5,\\
    &g_2(x)=\sum_{i=1}^n \overline{x}_i\mod 3,
\end{align*}
where $\overline{x}_i=\sum_{j=1}^i x_j\mod 2$.
At a high level, the sketch $f(x)$ is the main tool we use to approximately locate the error in $x$.
The count sketches $g_1(x)$ and $g_2(x)$ are added to allow us to detect how many markers are created or destroyed by the error, and to distinguish between the cases where there is no error or a transposition occurs.

With the above in mind, we define the preliminary code
\begin{equation*}
    \cC'=\left\{x\in\bits^n | (x_{n-3},\dots,x_n)=(0,0,1,1),f(x)=s_0,g_1(x)=s_1,g_2(x)=s_2,\forall i\in[\ell_x]: |z^x_i|\leq \Delta \right\}
\end{equation*}
for appropriate choices of $s_0,s_1,s_2$.
Taking into account all constraints, the choice of $\Delta$ and $m$, and Lemma~\ref{lem:shortsegment}, the pigeonhole principle implies that we can choose $s_0,s_1,s_2$ so that this code has at most
\begin{equation}\label{eq:redprelim}
    4+\log(10n\cdot \Delta\cdot m+1)+1+2+2+1=\log n + O(\log\log n)
\end{equation}
bits of redundancy.

However, it turns out that the constraints imposed in $\cC'$ are not enough to handle a deletion or a transposition.
Intuitively, the reason for this is that, in order to make use of the sketch $f(x)$ when decoding, we will need additional information both about the hashes of the segments of $x$ that were affected by the error and the hashes of the corresponding corrupted segments in the corrupted string $y$.
Therefore, given a vector $z^x$ and the hash function $h$ guaranteed by Lemma~\ref{lem:hash}, we will be interested in the associated \emph{hash multiset}
\begin{equation*}
    H_x=\{\{h(z^x_1),\dots,h(z^x_{\ell_x})\}\}
\end{equation*}
over $[m]$.
As we shall see, a deletion or transposition will change this multiset by at most $4$ elements.
Therefore, we will expurgate $\cC'$ so that any pair of remaining codewords $x$ and $x'$ satisfy either $H_x=H_{x'}$ or $|H_x\triangle H_{x'}|\geq 10$.
This will allow us to recover the true hash multiset of $x$ from the hash multiset of the corrupted string.
The following lemma shows that this expurgation adds only an extra $O(\log m)=O(\log\log n)$ bits of redundancy.
\begin{lem}\label{lem:multiset}
    There exists a code $\cC\subseteq \cC'$ of size
    \begin{equation*}
        |\cC|\geq \frac{|\cC'|}{m^{10}}
    \end{equation*}
    such that for any $x,x'\in\cC$ we either have $H_x=H_{x'}$ or $|H_x\triangle H_{x'}|\geq 10$.
\end{lem}
\begin{proof}
Let $\cS$ be the family of multisets over $[m]$ with at most $n$ elements.
Order the multisets $S$ in $\cS$ in decreasing order according to the number $N(S)$ of codewords $x\in\cC'$ such that $H_{x'}=S$.
The expurgation procedure works iteratively by considering the surviving multiset $S$ with the largest $N(S)$, removing all codewords $x\in\cC'$ associated to $S'\in\cS$ such that $S'\neq S$ and $|S\triangle S'|<10$, and updating the values $N(S)$ for $S\in \cS$.
Since there are at most $m^{10}$ multisets $S'$ satisfying the conditions above and $N(S)\geq N(S')$ for all such $S'$, we are guaranteed to keep at least a $\frac{1}{m^{10}}$-fraction of every subset of codewords considered in each round of expurgation.
This implies the desired result.
\end{proof}

We will take our \emph{error-locating} code to be the expurgated code $\cC$ guaranteed by Lemma~\ref{lem:multiset}.
By~\eqref{eq:redprelim} and the choice of $m$, it follows that there exists a choice of $s_0$ and $s_1$ such that $\cC$ has $\log n+O(\log\log n)$ bits of redundancy.
We prove the following result in Section~\ref{sec:errorlocation}, which states that, given a corrupted version of $x\in\cC$, we can identify a small interval containing the position where the error occurred.
\begin{thm}\label{thm:locate}
If $x\in\cC$ is corrupted into $y$ via one deletion or transposition, we can recover from $y$ a window $W\subseteq [n]$ of size $|W|\leq 10^{10}\log^4 n$ that contains the position where the error occurred (in the case of a transposition, we take the error location to be the smallest of the two affected indices).
\end{thm}

\subsection{Error correction from approximate error location}\label{sec:correctfromapprox}

In this section, we argue how we can leverage Theorem~\ref{thm:locate} to correct one deletion or one transposition by adding $O(\log\log n)$ bits of redundancy to $\cC$, thus proving Theorem~\ref{thm:deltrans}.

Let $L=10^{10}\log^4 n$.
We partition $[n]$ into consecutive disjoint intervals $B^{(1)}_1,B^{(1)}_2,\dots,B^{(1)}_t$ of length $2L+1$.
Moreover, we define a family of shifted intervals $B^{(2)}_1,\dots,B^{(2)}_{t-1}$ where $B^{(2)}_i=[a+L,b+L]$ if $B^{(1)}_i=[a,b]$.
For a given string $x\in\bits^n$, let $x^{(1,i)}$ denote its substring corresponding to $B^{(1)}_i$ and $x^{(2,i)}$ its substring corresponding to $B^{(2)}_i$.

The key property of these families of intervals we exploit is the fact that the window $W$ of length at most $L$ guaranteed by Theorem~\ref{thm:locate} satisfies either $W\subseteq B^{(1)}_i$ or $W\subseteq B^{(2)}_i$ for some $i$.
As a result, we are able to recover $x^{(1,j)}$ (resp.\ $x^{(2,j)}$) for all $j\neq i$ from $y$ if $W\subseteq B^{(1)}_i$ (resp.\ $W\subseteq B^{(2)}_i$).
Moreover, we can also recover a string $y^{(i)}$ that is obtained from $x^{(1,i)}$ or $x^{(2,i)}$ via at most one deletion or one transposition.
Therefore, it suffices to reveal an additional sketch which allows us to correct a deletion or a transposition in strings of length $2L+1=O(\log^4 n)$ for each interval.
Crucially, since we can already correctly recover all bits of $x$ except for those in the corrupted interval, we may XOR all these sketches together and only pay the price of one such sketch.
We proceed to discuss this more concretely.

Suppose $\widehat{f}:\bits^{2L+1}\to \bits^\ell$ is a sketch with the following property: \emph{If $z\in\bits^{2L+1}$ is transformed into $y$ via at most one deletion or one transposition, then knowledge of $y$ and $\widehat{f}(z)$ is sufficient to recover $z$ uniquely.}
It is easy to construct such a sketch with $\ell=O(\log L)=O(\log\log n)$~\cite{GYM18}.
For completeness, we provide an instantiation in Appendix~\ref{app:sketchdeltrans}.
Armed with $\widehat{f}$, we define the full sketches
\begin{equation*}
    \widehat{g}_1(x)=\bigoplus_{i=1}^t \widehat{f}(x^{(1,i)})
\end{equation*}
and
\begin{equation*}
    \widehat{g}_2(x)=\bigoplus_{i=1}^{t-1} \widehat{f}(x^{(2,i)})
\end{equation*}
Note that $\widehat{g}_b(x)$ has length $\ell=O(\log\log n)$ for $b\in\bits$.
Then, we take our final code to be
\begin{equation*}
    \widehat{\cC}=\{x\in\cC:\widehat{g}_1(x)=s_3,\widehat{g}_2(x)=s_4\}
\end{equation*}
which has redundancy $\log n+O(\log\log n)$ for some choice of $s_3$ and $s_4$.
To see that $\widehat{\cC}$ indeed corrects one deletion or one transposition, note that, by the discussion above, if the window $W$ guaranteed by Theorem~\ref{thm:locate} satisfies $W\subseteq B^{(1)}_i$, then we can recover $\widehat{f}(x^{(1,i)})$ from $\widehat{g}_1(x)$ and $y$, along with a string $y^{(i)}$ obtained from $x^{(1,i)}$ by at most one deletion or one transposition.
Then, the properties of $\widehat{f}$ ensure that we can uniquely recover $x^{(1,i)}$ from $y^{(i)}$ and the sketch $\widehat{f}(x^{(1,i)})$.
The reasoning for when $W\subseteq B^{(2)}_i$ is analogous.
This yields Theorem~\ref{thm:deltrans}.

\subsection{Proof of Theorem~\ref{thm:locate}}\label{sec:errorlocation}

We prove Theorem~\ref{thm:locate} in this section, which concludes our argument.
Fix $x\in\cC$ and suppose $y$ is obtained from $x$ via one deletion or one transposition.
We consider several independent cases based on the fact that a marker cannot overlap with itself, that we can identify whether a deletion occurred by computing $|y|$, and that we can identify whether a transposition occurred by comparing $g_2(x)$ and $g_2(y)$.

\subsubsection{Locating one deletion}\label{sec:del}

In this section, we show how we can localize one deletion appropriately.
Fix $x\in\cC$ and suppose that a deletion is applied to $z^x_i$.
The following lemma holds due to the marker structure.
\begin{lem}\label{lem:casesdel}
    A deletion either (i) Creates a new marker and does not delete any existing markers, in which case $\ell_y=\ell_x+1$, (ii) Deletes an existing marker and does not create any new markers, in which case $\ell_y=\ell_x-1$, or (iii) Neither deletes existing markers nor creates new markers, in which case $\ell_y=\ell_x$.
\end{lem}
\begin{proof}
Without loss of generality, we may assume that the deletion is applied to the first bit of a $0$-run or to the last bit of a $1$-run in $x\in\cC$.
The desired result is implied by the following three observations:
First, if the deletion is applied to a run of length at least $3$, then no marker is created nor destroyed.
Second, if the deletion is applied to a run of length $2$, then a marker may be destroyed, but no marker is created.
Finally, if the deletion is applied to a run of length $1$, then a marker may be created, but no marker is destroyed.
\end{proof}

Note that we can distinguish between the cases detailed in Lemma~\ref{lem:casesdel} by comparing $g_1(x)$ and $g_1(y)$.
Thus, we analyze each case separately:
\begin{enumerate}
    \item $\ell_y=\ell_x$: In this case, we have
    \begin{equation}\label{eq:zstruct1}
        z^y=(z^x_1,\dots,z^x_{i-1},z'_i,z^x_{i+1},\dots,z^x_{\ell_x}),
    \end{equation}
    where $z'_i$ is obtained from $z^x_i$ by a deletion (in particular, $|z'_i|=|z^x_i|-1$).
    Therefore, it holds that
    \begin{align*}
        f(x)-f(y)&=\sum_{j=1}^{\ell_x} j(|z^x_j|\cdot m+h(z^x_j))-\sum_{j=1}^{\ell_y} j(|z^y_j|\cdot m+h(z^y_j)) \mod L \\
        &=i(|z^x_i|\cdot m + h(z^x_i)-|z'_i|\cdot m-h(z'_i)) \\
        &=i(m+h(z^x_i)-h(z'_i)),
    \end{align*}
    where the second equality uses~\eqref{eq:zstruct1} and $\ell_y=\ell_x$.
    Let $H_y$ denote the hash multiset of $y$.
    Then, we know that $|H_x\triangle H_y|\leq 2$.
    Therefore, we can recover $H_x$ from $H_y$, which means that we can recover $h(z^x_i)-h(z'_i)$.
    Indeed, if $h(z^x_i)-h(z'_i)=0$ then $H_x=H_y$.
    On the other hand, if $h(z^x_i)-h(z'_i)\neq 0$ then $|H_x\triangle H_y|= 2$ and we recover both $h(z^x_i)$ (the element in $H_x$ but not in $H_y$) and $h(z'_i)$ (the element in $H_y$ but not in $H_x$).
    As a result, we know $m+h(z^x_i)-h(z'_i)$.
    Since it also holds that $m+h(z^x_i)-h(z'_i)\neq 0$ (because $|h(z^x_i)-h(z'_i)|<m$), we can recover $i$ from $f(x)-f(y)$.
    This gives a window $W$ of length at most $\Delta=O(\log n)$.
    
    \item $\ell_y=\ell_x-1$: In this case, the marker at the end of $z^x_i$ is destroyed, merging $z^x_i$ and $z^x_{i+1}$.
    Observe that if $i=\ell_x$ then we can simply detect that the last marker in $x$ was destroyed.
    Therefore, we assume that $i<\ell_x$, in which case we have
    \begin{equation}\label{eq:zstruct2}
        z^y=(z^x_1,\dots,z^x_{i-1},z'_i,z^x_{i+2},\dots,z^x_{\ell_x}),
    \end{equation}
    where $|z'_i|=|z^x_i|+|z^x_{i+1}|-1$.
    Consequently, it holds that
    \begin{align*}
        f(x)-f(y)&=\sum_{j=1}^{\ell_x} j(|z^x_j|\cdot m+h(z^x_j))-\sum_{j=1}^{\ell_y} j(|z^y_j|\cdot m+h(z^y_j)) \mod L \\
        &=i(|z^x_i|\cdot m+h(z^x_i))+(i+1)(|z^x_{i+1}|\cdot m+h(z^x_{i+1}))-i(|z'_i|\cdot m+h(z'_i))\\
        &+\sum_{j=i+2}^{\ell_x}(|z^x_j|\cdot m+h(z^x_j))\\
        &=\sum_{j=i+2}^{\ell_x}(|z^x_j|\cdot m+h(z^x_j))+i(m+h(z^x_i)+h(z^x_{i+1})-h(z'_i))\\
        &+(|z^x_{i+1}|\cdot m+h(z^x_{i+1})).
    \end{align*}
    Note that, since $|H_x\triangle H(y)|\leq 3$, we can recover $H_x$ from $H_y$.
In particular, this means that we know $h(z^x_i)+h(z^x_{i+1})-h(z'_i)$.
Therefore, for $i'=\ell_y-1,\ell_y-2,\dots,i$ we can compute the ``potential function''
\begin{align*}
    \ap(i')&=\sum_{j=i'+1}^{\ell_y}(|z^y_j|\cdot m+h(z^y_j))+i'(m+h(z^x_i)+h(z^x_{i+1})-h(z'_i))\\
    &=\sum_{j=i'+2}^{\ell_x}(|z^x_j|\cdot m+h(z^x_j))+i'(m+h(z^x_i)+h(z^x_{i+1})-h(z'_i)).
\end{align*}
Note that
\begin{equation}\label{eq:approx1}
    |\ap(i)-(f(x)-f(y))|=||z^x_{i+1}|\cdot m+h(z^x_{i+1})|\leq \Delta\cdot m+m\leq 10^7\log^2 n.
\end{equation}
Moreover, we also have
\begin{multline}\label{eq:monotone1}
    \ap(i'-1)-\ap(i')=|z^x_{i'+1}|\cdot m+h(z^x_{i'+1})-(m+h(z^x_i)+h(z^x_{i+1})-h(z'_i))\\\geq 4m-3m=m.
\end{multline}
This suggests the following procedure for recovering the window $W$.
Sequentially compute $\ap(i')$ for $i'$ starting at $\ell_y-1$ until we find $i^\star\geq i$ such that $|\ap(i')-(f(x)-f(y))|\leq 10^6\log^2 n$.
This is guaranteed to exist since $i'=i$ satisfies this property.
We claim that $i^\star-i\leq 10^7\log n$.
In fact, if this is not the case then the monotonicity property in~\eqref{eq:monotone1} implies that 
\begin{equation*}
    |\ap(i)-(f(x)-f(y))|> m\cdot 10^7\log n>10^7\log^2 n,
\end{equation*}
contradicting~\eqref{eq:approx1}.
Since $|z^x_j|\leq \Delta$ for every $j$, recovering $i^\star$ also yields a window $W\subseteq [n]$ of size
\begin{equation*}
    |W|=10^6\log n\cdot\Delta=10^9\log^2 n
\end{equation*}
containing the error position, as desired.

    \item $\ell_y=\ell_x+1$: This case is similar to the previous one.
    We present it for completeness.
    The deletion causes the segment $z^x_i$ to be split into two consecutive segments $z'_i$ and $z''_i$ such that $|z'_i|+|z''_i|=|z^x_i|-1$.
    Therefore, we have
    \begin{equation}\label{eq:zstruct3}
        z^y=(z^x_1,\dots,z^x_{i-1},z'_i,z''_i,z^x_{i+1},\dots,z^x_{\ell_x}).
    \end{equation}
    We may compute
    \begin{align*}
        f(x)-f(y)&=\sum_{j=1}^{\ell_x} j(|z^x_j|\cdot m+h(z^x_j))-\sum_{j=1}^{\ell_y} j(|z^y_j|\cdot m+h(z^y_j)) \mod L \\
        &=i(|z^x_i|\cdot m+h(z^x_i))-i(|z'_i|\cdot m+h(z'_i))-(i+1)(|z''_i|\cdot m+h(z''_i))\\
        &-\sum_{j=i+1}^{\ell_x}(|z^x_j|\cdot m+h(z^x_j))\\
        &=-\sum_{j=i+1}^{\ell_x}(|z^x_j|\cdot m+h(z^x_j))+i(m+h(z^x_i)-h(z'_i)-h(z''_i))\\
        &-(|z''_i|\cdot m+h(z''_i)).
    \end{align*}
    As in the previous case, we can recover $H_x$ from $H_y$, and this implies we can also recover $h(z^x_i)-h(z'_i)-h(z''_i)$.
    Therefore, for $i'\geq i$ we can compute
    \begin{align*}
        \ap(i')&=-\sum_{j=i'+2}^{\ell_y}(|z^y_j|\cdot m+h(z^y_j))+i'(m+h(z^x_i)-h(z'_i)-h(z''_i))\\
        &=-\sum_{j=i'+1}^{\ell_x}(|z^x_j|\cdot m+h(z^x_j))+i'(m+h(z^x_i)-h(z'_i)-h(z''_i)).
    \end{align*}
    Then,
    \begin{equation}\label{eq:approx2}
        |\ap(i)-(f(x)-f(y))|=||z''_i|\cdot m+h(z''_i)|\leq \Delta\cdot m+m\leq 10^7\log^2 n,
    \end{equation}
    since $|z''_i|\leq |z^x_i|\leq \Delta$.
    Furthermore, for $i'>i$ we have
    \begin{multline}\label{eq:monotone2}
        \ap(i')-\ap(i'-1)=|z^x_{i'}|\cdot m+h(z^x_{i'})+(m+h(z^x_i)-h(z'_i)-h(z''_i))\\\geq 4m-2m=2m.
    \end{multline}
    As in the previous case, we can exploit~\eqref{eq:approx2} and~\eqref{eq:monotone2} to recover an appropriate window $W\subseteq [n]$ of size at most $10^9\log^2 n$.
\end{enumerate}

\subsubsection{Locating one transposition}\label{sec:trans}

In this section, we show how we can localize one transposition appropriately.
Fix $x\in\cC$ and suppose that a transposition is applied with the left bit in $z^x_i$ (note the right bit may be in $z^x_{i+1}$).
Then, the following lemma holds.
\begin{lem}\label{lem:casestrans}
    A transposition either (i) Creates a new marker and does not delete any existing markers, in which case $\ell_y=\ell_x+1$, (ii) Deletes an existing marker and does not create any new markers, in which case $\ell_y=\ell_x-1$, (iii) Neither deletes existing markers nor creates new markers, in which case $\ell_y=\ell_x$, (iv) Deletes two existing consecutive markers and does not create any new markers, in which case $\ell_y=\ell_x-2$, or (v) Creates two consecutive new markers but does not delete any existing markers, in which case $\ell_y=\ell_x+2$.
\end{lem}
\begin{proof}
We obtain the desired statement via case analysis.
If the leftmost bit of the adjacent transposition belongs to a run of length at least $3$ in $x$, then no marker is created and at most one marker is destroyed, and likewise for the case where the leftmost bit belongs to some $0$-run of length $2$.
On the other hand, the leftmost bit belongs to a $1$-run of length $2$, then no marker is created, but at most two markers may be destroyed (consider applying one transposition to the underlined bits in $001\underline{10}011$).
If the leftmost bit belongs to a $0$-run of length $1$, then no marker is destroyed and at most two consecutive markers may be created (consider applying one transposition to the underlined bits in $001\underline{01}011$).
Finally, if the leftmost bit belongs to a $1$-run of length $1$, then no marker is destroyed and at most one marker is created (consider applying one transposition to the underlined bits in $0\underline{10}1$).
\end{proof}

As before, we can distinguish between the cases detailed in Lemma~\ref{lem:casestrans} by comparing $g_1(x)$ and $g_1(y)$.
Cases (i), (ii), and (iii) in Lemma~\ref{lem:casestrans} are analogous to the respective cases considered for a deletion in Section~\ref{sec:del}.
Therefore, we focus on cases (iv) and (v).

\begin{enumerate}

    \item $\ell_y=\ell_x$: In this case, we have
    \begin{equation*}
        f(x)-f(y) = i(h(z^x_i)-h(z'_i)),
    \end{equation*}
    and we can recover $i$ by first recovering $h(z^x_i)-h(z'_i)\neq 0$, which holds because $z'_i$ is obtained from $z^x_i$ via one transposition.
    
    \item $\ell_y=\ell_x-1$: In this case, we have
    \begin{equation*}
        f(x)-f(y)=\sum_{j=i+2}^{\ell_x}(|z^x_j|\cdot m+h(z^x_j))+i(h(z^x_i)+h(z^x_{i+1})-h(z'_i))
        +(|z^x_{i+1}|\cdot m+h(z^x_{i+1})),
    \end{equation*}
    and we can then use the exact same approach as in Case (ii) from Section~\ref{sec:del}.
    
    \item $\ell_y=\ell_x+1$: In this case, we have
    \begin{equation*}
        f(x)-f(y)=-\sum_{j=i+1}^{\ell_x}(|z^x_j|\cdot m+h(z^x_j))+i(h(z^x_i)-h(z'_i)-h(z''_i))-(|z''_i|\cdot m+h(z''_i)),
    \end{equation*}
    and we can then use the exact same approach as in Case (iii) from Section~\ref{sec:del}.

    \item $\ell_y=\ell_x-2$: In this case, two consecutive markers are deleted and no new markers are created, so $z^x_i$, $z^x_{i+1}$, and $z^x_{i+2}$ are merged into a corrupted segment $z'_i$ satisfying $|z'_i|=|z^x_i|+|z^x_{i+1}|+|z^x_{i+2}|$.
    In general, we have
    \begin{equation*}
        z^y=(z^x_1,\dots,z^x_{i-1},z'_i,z^x_{i+3},\dots,z^x_{\ell_x}),
    \end{equation*}
    and so
    \begin{align*}
        f(x)-f(y)&=\sum_{j=1}^{\ell_x} j(|z^x_j|\cdot m+h(z^x_j))-\sum_{j=1}^{\ell_y} j(|z^y_j|\cdot m+h(z^y_j)) \mod L\\
        &=2\sum_{j=i+3}^{\ell_x}(|z^x_j|\cdot m+h(z^x_j))+i(h(z^x_i)+h(z^x_{i+1})+h(z^x_{i+2})-h(z'_i))\\
        &+(|z^x_{i+1}|\cdot m+h(z^x_{i+1}))+2(|z^x_{i+2}|\cdot m+h(z^x_{i+2})).
    \end{align*}
    Since $|H_x\triangle H_y|\leq 4$, we can recover $H_x$ from $H_y$, which implies that we can recover $h(z^x_i)+h(z^x_{i+1})+h(z^x_{i+2})-h(z'_i)$.
    As before, this means that for $i'\geq i$ we can compute the potential function
    \begin{align*}
        \ap(i')&=2\sum_{j=i'+1}^{\ell_y}(|z^y_j|\cdot m+h(z^y_j))+i'(h(z^x_i)+h(z^x_{i+1})+h(z^x_{i+2})-h(z'_i))\\
        &=2\sum_{j=i'+3}^{\ell_x}(|z^x_j|\cdot m+h(z^x_j))+i'(h(z^x_i)+h(z^x_{i+1})+h(z^x_{i+2})-h(z'_i)).
    \end{align*}
    Exploiting the fact that
    \begin{align*}
        |\ap(i)-(f(x)-f(y))|&= (|z^x_{i+1}|\cdot m+h(z^x_{i+1}))+2(|z^x_{i+2}|\cdot m+h(z^x_{i+2}))\\
        &\leq 3(\Delta\cdot m+m)\\
        &\leq 10^{10}\log^2 n
    \end{align*}
    and
    \begin{equation*}
        \ap(i'-1)-\ap(i')\geq 8m - 3m=5m,
    \end{equation*}
    we can use the approach from Section~\ref{sec:del} to recover the relevant window $W\subseteq [n]$ of size at most $10^{10}\log^3 n$ containing the error position in $y$.
    
    \item $\ell_y=\ell_x+2$: This case is similar to the previous one.
    Two consecutive markers are created and none are deleted, meaning that $z^x_i$ is transformed into two consecutive corrupted segments $z'_i$, $z''_i$, and $z'''_i$.
    Therefore,
    \begin{equation*}
        z^y=(z^x_1,\dots,z^x_{i-1},z'_i,z''_i,z'''_i,z^x_{i+1},\dots,z^x_{\ell_x})
    \end{equation*}
    with $|z^x_i|=|z'_i|+|z''_i|+|z'''_i|$.
    We have
    \begin{align*}
        f(x)-f(y)&=\sum_{j=1}^{\ell_x} j(|z^x_j|\cdot m+h(z^x_j))-\sum_{j=1}^{\ell_y} j(|z^y_j|\cdot m+h(z^y_j)) \mod L\\
        &=-2\sum_{j=i+1}^{\ell_x}(|z^x_j|\cdot m+h(z^x_j))+i(h(z^x_i)-h(z'_i)-h(z''_i)-h(z'''_i))\\
        &-(|z''_i|\cdot m+h(z''_i))-2(|z'''_i|\cdot m+h(z'''_i)).
    \end{align*}
    As above, we can recover $H_x$ from $H_y$ and thus also recover $h(z^x_i)-h(z'_i)-h(z''_i)-h(z'''_i)$.
    Consequently, for $i'\geq i$ we can compute the potential function
    \begin{align*}
        \ap(i')&=-2\sum_{j=i'+3}^{\ell_y}(|z^y_j|\cdot m+h(z^y_j))+i'(h(z^x_i)-h(z'_i)-h(z''_i)-h(z'''_i))\\
        &=-2\sum_{j=i'+1}^{\ell_x}(|z^x_j|\cdot m+h(z^x_j))+i'(h(z^x_i)-h(z'_i)-h(z''_i)-h(z'''_i)).
    \end{align*}
    Since
    \begin{align*}
        |\ap(i)-(f(x)-f(y))|&= (|z''_i|\cdot m+h(z''_i))+2(|z'''_i|\cdot m+h(z'''_i))\\
        &\leq 3(\Delta\cdot m+m)\\
        &\leq 10^{10}\log^2 n
    \end{align*}
    and
    \begin{equation*}
        \ap(i')-\ap(i'-1)\geq 8m - 3m=5m,
    \end{equation*}
    we can follow the previous approach to recover a window $W\subseteq [n]$ of size at most $10^{10}\log^3 n$ containing the error position.
\end{enumerate}

\section{Open problems}

Our work leaves open several natural avenues for future research.
We highlight a few of them here:
\begin{itemize}

    \item Given the effectiveness of weighted VT sketches in the construction of nearly optimal non-binary single-edit correcting codes in Section~\ref{sec:nonbinedit} with fast encoding and decoding, it would be interesting to find further applications of this notion.
    
    \item We believe that the code we introduce and analyze in Section~\ref{sec:delandsub2} is actually uniquely decodable under one deletion and one substitution.
    Proving this would be quite interesting, since then we would also have explicit uniquely decodable single-deletion single-substitution correcting codes with redundancy matching the existential bound, analogous to what is known for two-deletion correcting codes~\cite{guruswami2021explicit}.
    
    \item The code we designed in Section~\ref{sec:deltrans} fails to correct an arbitrary substitution.
    Roughly speaking, the reason behind this is that one substitution may simultaneously destroy and create a marker with a different starting point.
    As the clear next step, it would be interesting to show the existence of a binary code correcting one \emph{edit} error or one transposition with redundancy $\log n+O(\log\log n)$.
\end{itemize}

\bibliographystyle{alpha}
\bibliography{reference}

\appendix

\section{Naive sketch for one deletion or one transposition}\label{app:sketchdeltrans}

In this section, we provide a concrete instantiation of the sketch $\widehat{f}$ used in Section~\ref{sec:correctfromapprox} which is implicit in~\cite{GYM18}.
Let $L'=2L+1$.
We claim that we may take $\widehat{f}:\bits^{L'}\to\bits^\ell$ of the form
\begin{equation*}
    \widehat{f}(z)=\mathsf{bin}\left(\sum_{i=1}^n iz_i \mod (L'+1), \sum_{i=1}^n i\overline{z}_i \mod (2L'+1)\right),
\end{equation*}
where $\mathsf{bin}$ denotes binary expansion up to $\lceil \log(2L'+1)\rceil$ bits and $\overline{z}_i=\sum_{j=1}^i z_j\mod 2$.
Note that in this case $\ell = O(\log L)$, as desired.

It remains to see that $\widehat{f}$ above satisfies the desired property.
Suppose that $y$ is obtained from $z\in\bits^{L'}$ via at most one deletion or one transposition.
Our goal is to show that we can determine $z$ uniquely from $y$ and $\widehat{f}(z)$.
First, note that we can detect if a deletion occurred by computing $|y|$.
If $|y|=L'-1$, then, as shown by Levenshtein~\cite{Lev65}, we can use $y$ and the first part of $\widehat{f}(z)$ to recover $z$.
Else, if $|y|=L'$, then we observe that an adjacent transposition in $z$ is equivalent to a substitution in $\overline{z}$.
Therefore, as shown as well by Levenshtein~\cite{Lev65}, we can use $y$ and the second part of $\widehat{f}(z)$ to recover $z$ since there is a unique correspondence between $z$ and $\overline{z}$.

\end{document}